\documentclass{article}

\usepackage{arxiv}
\usepackage[numbers]{natbib}
\usepackage[utf8]{inputenc} 
\usepackage[T1]{fontenc}    
\usepackage{hyperref}       
\usepackage{url}            
\usepackage{booktabs}       
\usepackage{amsfonts}       
\usepackage{amssymb,amsmath,amsthm}
\usepackage{latexsym,bm}
\usepackage{nicefrac}       
\usepackage{microtype}      
\usepackage{lipsum}
\usepackage{graphicx}
\graphicspath{ {./images/} }
\usepackage{url}
\usepackage{xcolor}
\usepackage{subcaption}
\usepackage{caption}
\usepackage{algorithm}
\usepackage{algpseudocode}
\usepackage{multirow}

\newtheorem{theorem}{Theorem}[section]
\theoremstyle{definition}
\newtheorem{remark}{Remark}[section]

\title{Neumann Series-based Neural Operator for Solving Inverse Medium Problem}

\author{
Ziyang Liu \\
  Department of Mathematical Sciences\\
  Tsinghua University\\
  Beijing, 100084, China \\
  \texttt{zy-iu20@mails.tsinghua.edu.cn} \\
   \And
Fukai Chen \\
  Department of Mathematical Sciences\\
  Tsinghua University\\
  Beijing, 100084, China \\
  \texttt{cfk19@mails.tsinghua.edu.cn} \\
   \And
Junqing Chen \\
  Department of Mathematical Sciences\\
  Tsinghua University\\
  Beijing, 100084, China \\
  \texttt{jqchen@tsinghua.edu.cn} \\
   \And
Lingyun Qiu \\
  Yau Mathematical Sciences Center\\
  Tsinghua University\\
  Beijing, 100084, China \\
  \texttt{lyqiu@tsinghua.edu.cn} \\
   \And
Zuoqiang Shi \\
  Yau Mathematical Sciences Center\\
  Tsinghua University\\
  Beijing, 100084, China \\
  \texttt{zqshi@tsinghua.edu.cn} \\
   \And
}
\date{}

\begin{document}
\maketitle
\begin{abstract}
The inverse medium problem, inherently ill-posed and nonlinear, presents significant computational challenges. This study introduces a novel approach by integrating a Neumann series structure within a neural network framework to effectively handle multiparameter inputs. Experiments demonstrate that our methodology not only accelerates computations but also significantly enhances generalization performance, even with varying scattering properties and noisy data. The robustness and adaptability of our framework provide crucial insights and methodologies, extending its applicability to a broad spectrum of scattering problems. These advancements mark a significant step forward in the field, offering a scalable solution to traditionally complex inverse problems.
\end{abstract}


\section{Introduction}
\label{introduction}

\hspace{2em}The inverse medium problem is a significant mathematical challenge in physics and engineering. It involves determining the internal characteristics of a medium by observing wave scattering, including light, sound, and other electromagnetic waves \cite{colton2019inverse, ito2012direct, bao2005inverse1}. This problem is essential in fields where direct measurement is impractical. In medical diagnostics, techniques such as ultrasound imaging \cite{carevic2018solving} and MRI \cite{spencer2020tutorial} enable non-invasive visualization of internal body structures for accurate diagnosis. In geophysical exploration \cite{virieux2009overview, fichtner2013multiscale}, analysis of seismic wave scattering facilitates the detection of subsurface resources, such as oil and minerals.

\begin{figure*}[h!]
	\centering
	\includegraphics[width=0.3\textwidth]{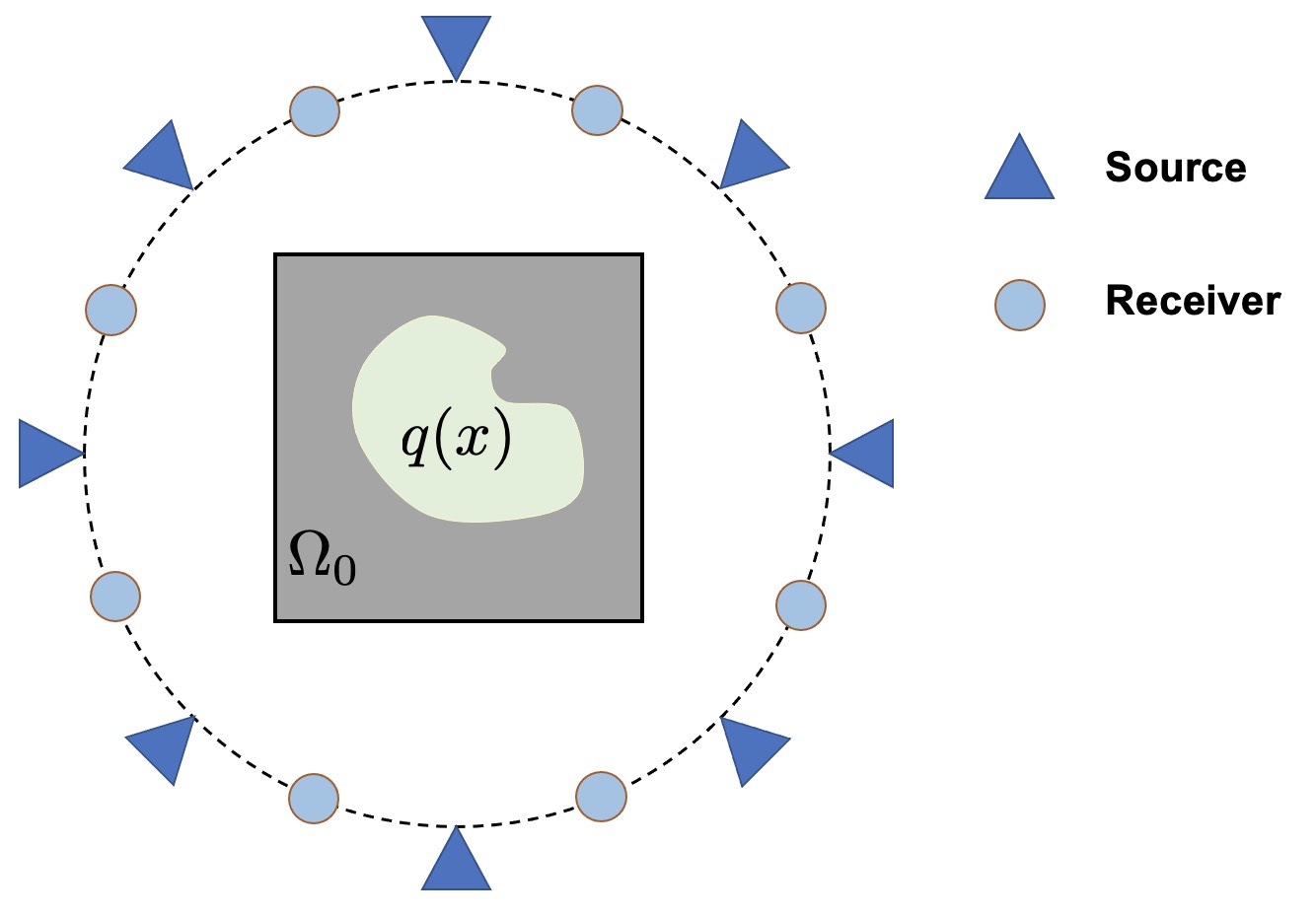}
	\caption{Schematic diagram of the inverse medium problem. }
	\label{Inverse}
\end{figure*}
This article focuses on a specific case of the two-dimensional inverse medium problem in penetrable media. For simplicity, and without loss of generality, we assume the incident field is a plane wave, represented by $u^i(x) = \exp(ik x\cdot d)$, where $k$ and $d$ denote the wavenumber and the direction of incidence, respectively. As depicted in Figure \ref{Inverse}, the perturbation of the refractive index, $q(x)$, also referred to as the scatterer, is compactly supported within the medium domain $\Omega_0$. Plane wave sources from various directions penetrate the medium, and receivers are uniformly positioned along the circular boundary to record the scattered field data. The inverse problem seeks to retrieve the parameter $q$ from near-field measurements \cite{bao2007inverse}. 

As an indispensable component for understanding the scattering mechanisms, the forward problem is defined as follows:

\begin{equation}
\label{forward}
\left\{	
\begin{aligned}
& \Delta u^s+k^2 (1+q) u^s= - k^2 qu^i \quad \mbox{in} \; \mathbb{R}^2, \\
& \lim _{r \rightarrow \infty} \sqrt{r}\left(\frac{\partial u^s}{\partial r}-i k u^s\right)=0.
\end{aligned}
\right .	
\end{equation}
Colton et al.\cite{colton2019inverse} have proven the existence and uniqueness of solution for the forward equation. Considering that $u^s$ linearly depends on the right-hand side term $qu^i$, the forward map can be defined as:

\begin{equation}
\label{forward map}
u^s = \mathcal{S}(q)(qu^i),
\end{equation}
where $\mathcal{S}$ is nonlinear with respect to $q$. 

Due to the inherent challenges posed by nonlinearity and unbounded domains, solving forward problems requires sophisticated numerical approaches. Among these, the imposition of artificial boundary conditions, such as the Perfectly Matched Layer (PML) \cite{berenger1994perfectly,johnson2021notes,cl2005} and Absorbing Boundary Condition (ABC) \cite{higdon1987numerical,liu2018improved}, is widely adopted. The forward problem is first transformed into a manageable problem in a bounded domain and then solved by the Finite Difference Method (FDM) and Finite Element Method (FEM). Additionally, Vainikko \cite{vainikko2000fast} and Hohage \cite{hohage2006fast} addressed the problem by solving the equivalent Lippmann–Schwinger integration equation. Kirsch and Monk \cite{1994An} advanced this field by coupling the FEM with the boundary integral method. Colton and Kress summarize the integral equation methods in \cite{colton2013integral}.

Compared with forward problems, inverse medium problems are more ill-posed and require robust numerical solver. The decomposition method \cite{colton1988inverse,colton1991approximation} and the direct sampling method \cite{cakoni2002linear} are employed for rapid recovery. For detailed reconstructions, the most commonly employed method involves iterative reconstruction techniques, which address the following optimization problem:

\begin{equation}
\min_{q} J(q) = \sum_{j=1}^{M} J_j(q) = \sum_{j=1}^{M} \frac{1}{2}\| T \mathcal{S}(q)(q u_j^i) - \bm{d_j}\|_2^2,
\label{inverse_set}
\end{equation}
where  $\bm{d_j}\in \mathbb{C}^N,j=1,2,\ldots,M$ is the collected scattered data, while $M$ and $N$ denote the number of sources and receivers, respectively. Minimizing $J(q)$ aims to adjust the scatterer distribution such that the simulated data aligns with measurements. A significant limitation, however, is the need to solve the forward problems in each iteration, which considerably increases computational overhead. Furthermore, in practical applications, the use of high-order optimization methods is crucial for faster convergence. These optimization methods always depend on the gradient of $J(q)$, typically obtained by using the adjoint state method. Nevertheless, the demand of solving an additional adjoint state equation further increases the computational costs. It is noteworthy that \eqref{inverse_set} is defined for a single wavenumber, which establishes the foundation for formulating objective functions in multi-frequency measurement scenarios \cite{bao2015inverse}. For simplicity 
the standard regularization term typically included in such optimization problems has been omitted.

Recent advances have leveraged machine learning techniques to address PDE-based physical problems. The most straightforward method adopts an end-to-end learning approach, establishing direct correlations between collected data and input parameters. For instance, SwitchNet \cite{khoo2019switchnet} employs a low-rank representation to solve the far-field operator, and InversionNet \cite{wu2019inversionnet} utilizes a convolutional neural network (CNN) to directly derive the seismic inversion operator. However, many of these approaches treat inverse problems as entirely black-box processes, overlooking the role of PDEs(namely, the physical background). Consequently, these methods often lack interpretability, generalization, and robustness. To address the stability issue, recent studies \cite{guo2019supervised, rasht2022physics} have proposed to integrate the aforementioned optimization framework and neural network method, while merely employing neural networks as an alternative solver for the forward problem. This methodology allows the training stage to function as an offline process, enabling the online optimization framework to execute effectively with an approximated solution provided by the networks. Another significant advantage is that automatic differentiation \cite{paszke2017automatic} with neural networks can avoid the need to derive adjoint equations. For a more detailed summary of existing neural network methods, please refer to \cite{chen2020review}.

Taking \eqref{inverse_set} as an example, solving PDE-based inverse problems within the optimization framework necessitates solving a series of forward PDEs for varying parameter inputs. From a mathematical perspective, we aim for neural networks to approximate a mapping from parameter space to solution space, a process termed operator learning \cite{boulle2023mathematical} for parametric PDEs \cite{huang2022meta}. This task is typically nonlinear and ill-posed. Deep Operator Network (DeepONet) \cite{lu2021learning} addresses this issue by introducing branch and trunk nets to capture the low-rank relationships between input and output functions. In Fourier Neural Operator  \cite{li2020fourier} (FNO), the Fourier transformation allows the network to learn a parametrized form of the kernel in the spectral domain by applying pointwise multiplication in Fourier space. Deep Green Network (DGN) \cite{xu2021survey} directly learns the kernel in physical space. Beyond directly learning mappings between function spaces, another strategy employs Physics-Informed Neural Network (PINN) \cite{karniadakis2021physics, li2023physicsinformed} to learn solutions at a fixed set of parameters as a preliminary step. Subsequently, these solutions are refined through transfer learning \cite{goswami2022deep} or meta-learning \cite{huang2022meta} techniques.

Returning to the physics problem addressed in this article, it should be noted that the parameters of the forward problem originate from two distinct spaces. Current neural network methodologies typically interpret multi-parameters as multidimensional stacked channels, thereby neglecting their inherent physical distinctions. Moreover, the nonlinear characteristics of these parameters pose challenges to the generalization capabilities of these networks and complicate the achievement of satisfactory accuracy and robustness. In \cite{chen2024nsno}, we addressed this issue by employing the Neumann series, a skip connection method \cite{gilton2019neumann} commonly used to decouple multi-parameter inputs. The main contributions of this paper are summarized as follows:
\begin{itemize}

\item The mathematical foundations of the Neumann series embedding are analyzed in detail. Our previously proposed Neumann Series Neural Operator (NSNO) has been refined and designated as
a new explicit approach. 
\item Extensive numerical experiments validate the efficacy of our methods. In various practical scenarios, our approach significantly outperforms non-embedded networks, 
and the computational speed is 8 to 20 times comparing to the existing methods.
\item Robust generalization capabilities are demonstrated against out-of-distribution parameters, and the enhancement provided by the Neumann series proves to be compatible with numerous existing neural networks .
\end{itemize}

The structure of this paper is as follows. In Section 2, we provides an overview of the Neumann series. In Section 3, the integration about Neumann series and neural networks is discussed. Section 4 details experiments on network performance for forward problems. Section 5 exhibits the results of inverse reconstruction experiments. 

Finally, in Section 6 we draw some conclusion, discuss the findings and outline some future research directions.

\section{Neumann series}

\hspace{2em}The Neumann series is an important concept in functional analysis and linear algebra. It refers to a series used primarily for finding the inverse of a matrix or operator, especially in contexts where direct inversion is complex or not feasible. It is analogous to the geometric series in linear algebra and can be expressed in the form 
\begin{equation}
    (I-A)^{-1} = I+A + A^2+A^3+\dots,
\label{NS_base}
\end{equation}
where $A$ is a matrix or linear operator and $I$ is the identity. This series converges under the condition that the norm of $A$ is less than one. In this section, we explore the application of the Neumann series to devise an iterative solution format for the Helmholtz equation. 

We start with the Helmholtz equation with constant coefficients. Assume $\Omega$ is a bounded domain in $\mathbb{R}^2$, and the function $f$ is supported in $\Omega$, we consider

\begin{equation}\label{Simple}
\left\{
\begin{aligned}
& \Delta u+k^2 u= - k^2 f  \quad \text{in}\; \mathbb{R}^2\\
& \lim _{r \rightarrow \infty} \sqrt{r}\left(\frac{\partial u}{\partial r}-i k u\right)=0.
\end{aligned}
\right. 
\end{equation}
It is known that the solution to this equation can be represented by the following integral equation,
\begin{equation}
\begin{aligned}
u = \hat{\mathcal{S}}(f) &= -k^2\int_{\Omega} G(x,y) f(y) \, dy, \quad \forall x \in \mathbb{R}^2,\\
\text{where} \quad G(x,y) &= -\frac{i}{4} H_0^{(1)}(k|x-y|), \quad x\neq y,\,x,y\in \mathbb{R}^2,
\end{aligned}
\label{integration}
\end{equation}
and $H_0^{(1)}$ denotes the Hankel function of the first kind of order zero. This equation guarantees the existence and uniqueness of the solution. Similarly, the forward equation has an equivalent representation known as the Lippmann-Schwinger equation,

\begin{equation*}
	u^s(x) + k^2 \int_{\Omega} G(x,y) q(y) u^s(y) \, dy = -k^2 \int_{\Omega} G(x,y) q(y)u^i(y) \, dy.
\end{equation*}

By leveraging the operator $\hat{\mathcal{S}}$ defined in the integral equation \eqref{integration}, we can reformulate the integral equation as
\begin{equation}
	u^s = \hat{\mathcal{S}}(qu^s + qu^i),  \Rightarrow u^s = (I - \hat{\mathcal{S}}(q \cdot))^{-1}\hat{\mathcal{S}}(qu^i),
\label{LS}
\end{equation}
if $I-\hat{\mathcal{S}}(q\cdot)$ is invertible, 
where $\hat{\mathcal{S}}(q\cdot)$ performs the role of linear operator. By substituting \eqref{NS_base} into \eqref{LS}, we can derive that

\begin{equation}
\label{derive}
\begin{aligned}
	u^s &= \hat{\mathcal{S}}(qu^s + qu^i) \\
	&= \hat{\mathcal{S}}(qu^i) + \hat{\mathcal{S}}(q(\hat{\mathcal{S}}(qu^s+qu^i))\\
	&= \hat{\mathcal{S}}(qu^i) + \hat{\mathcal{S}}(q(\hat{\mathcal{S}}(qu^i)) + \hat{\mathcal{S}}(q(\hat{\mathcal{S}}(q(\hat{\mathcal{S}}(qu^s + qu^i))))) =  \dots,
\end{aligned}
\end{equation}
Thus, the nonlinearity can be handled by separate iteration subprocesses which involves the computation of the same linear operator,
\begin{equation}
\label{NS}
u^{(0)} = u^i, \quad u^{(j+1)} \triangleq \hat{\mathcal{S}}(qu^{(j)}), j=0,1,2,\cdots,
\end{equation}
As a result, we can use finite summation $U^{(L)} \triangleq \sum_{j = 1}^L u^{(j)}$ as an approximated solution to the forward problems \eqref{forward}.

Comparing with the original Neumann Series \eqref{NS_base}, the operator $\hat{\mathcal{S}} (q\cdot )$ plays a role similar to $A$. Therefore, the convergence of the series in \eqref{derive} is equivalent to its norm being less than one. Bao et al. \cite{bao2000regularity}  have shown that the operator $\hat{\mathcal{S}}$ is a bounded linear operator from $L^2(\Omega)$ to $L^2(\Omega)$. Suppose we have 
\[
\| \hat{\mathcal{S}} \|_{L^2{(\Omega)}} \leq C_1(k,\Omega),
\]
when $\| q\|_{L^\infty(\Omega)} < 1/(2C_1(k,\Omega)$, we can derive that, $\forall f \in {L^2(\Omega)}$,

\[
	\|\hat{\mathcal{S}} (qf) \|_{L^2(\Omega)} \leq C_1(k,\Omega) \| qf \|_{L^2(\Omega)}\leq C_1(k,\Omega) \| q\|_{L^\infty(\Omega)}\| f\|_{L^2(\Omega)} < \|f\|_{L^2(\Omega)}.
\]
Hence, the norm of the operator is confined within 1, which implies that the series $\{U^{(L)}\}_{N=0}^{\infty}$ converges exclusively to $u^s$. Although convergence of the corresponding Neumann series is only proven under conditions of weak scattering, this approach effectively transforms a nonlinear equation into a series of linear equations. This transformation not only preserves the underlying structure but also provides valuable insights into parameter embedding. And it inspire us to propose the neural network structure to solve the Helmholtz equation in the next sections.

\section{Network architecture}
\label{Forward Solver Architecture}

\hspace{2em}In forward problem \eqref{forward}, the scatterer $q(x)$ and the incident wave $u^i$ lie in different function spaces. The Neumann series provides a robust framework for decoupling the two parameters and linearizing the dependence on the parameters. On the one hand, globally mapping the scatterer information into each iterative step has been proven to significantly enhance the accuracy of networks solving the forward problem. On the other hand, the structure of the Neumann series, functioning as a series of linear mappings, offers a guiding rationale for implementing each level with a similarly structured subnetwork. We design two approaches to integrate the Neumann series into neural network architectures which will be detailed in subsequent sections.

\subsection{Implicit approach}
\label{Implicit approach}
\hspace{2em}We aim to train a neural operator that maps from the scatterer and incident wave function spaces to the scattered wave function space.
Based on our previous work \cite{chen2024nsno}, the network's architecture, depicted in Figure \ref{Implicit}(a), employs an implicit embedding approach analogous to a truncated Neumann series with $L$ terms truncation. We introduced a series of $L$ distinct subnetworks, $\{\mathcal{N}_{\theta_1}, \mathcal{N}_{\theta_2}, \cdots, \mathcal{N}_{\theta_L}\}$, each equipped with learnable parameters $\{\theta_1, \cdots, \theta_L\}$, collectively denoted as $\boldsymbol{\theta}$. This structure effectively positions the inputs $q(x)$ and the incident wave $u^i$ within different components of the network. Information from the scatterer undergoes a Hadamard product operation with the output from one subnetwork before entering the next, mirroring the process of the Neumann series, as described in \eqref{NS}. The outputs of the subnetworks are subsequently aggregated to produce the final network output.

Theoretically, if the Neumann series converges, the norms of inputs for each equation should gradually decrease. Inspired by this intuition, the implicit network architecture may experience a similar trend. To enhance the training stability for inputs across different scales and preserve the linearity of the target operator $\hat{\mathcal{S}}$, we provide a normalization strategy. As depicted in Figure \ref{Implicit}(b), the inputs are scaled by their $L^{\infty}$ norm before being processed by the subnetwork, and the output is subsequently rescaled to match the recorded magnitude. 

\begin{figure}[h!]
\centering
	\includegraphics[width=0.8\textwidth]{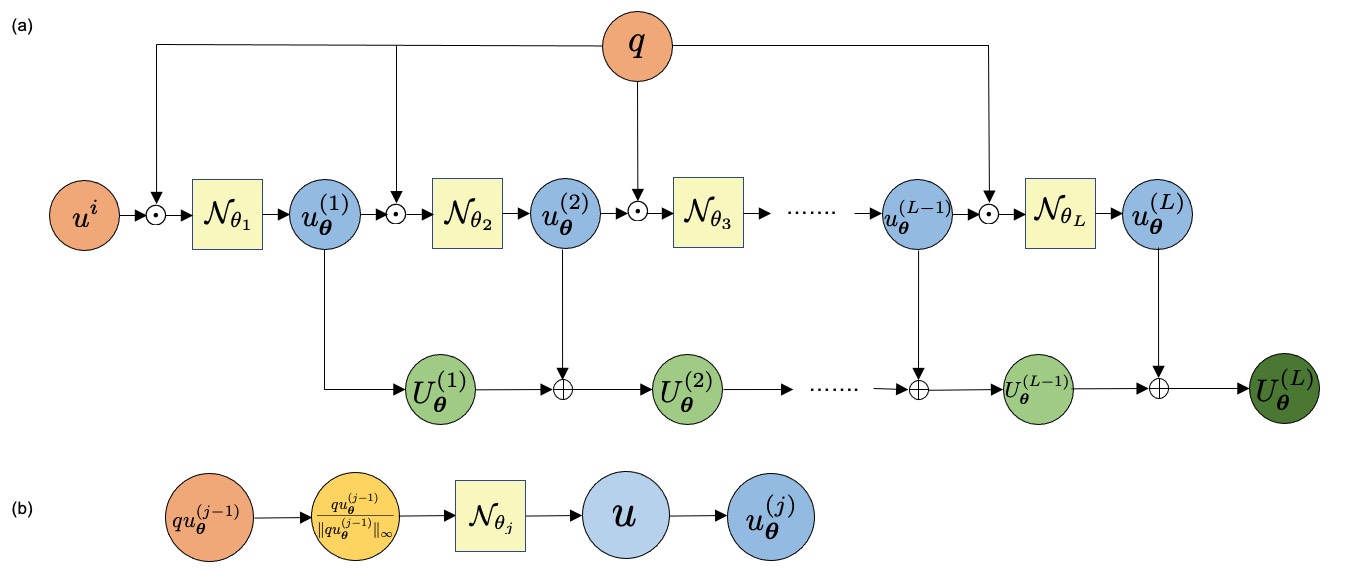}
	\caption{Network structure for the implicit approach. \textnormal{(a) $\odot$ refers to the element-by-element multiplication and $\oplus$ denotes the element-wise addition.} (b) detail structure of the normalizing operation. $u^i = u_{\boldsymbol{\theta}}^{(0)}, u_{\boldsymbol{\theta}}^{(j)} = u\| qu_{\boldsymbol{\theta}}^{(j-1)}\|_{\infty}, j = 1,2,\cdots,L$.} 
	\label{Implicit}
\end{figure}
 
To enhance the network's approximation capabilities, we designate broadly applicable neural operators, such as FNO \cite{li2020fourier}, CNO \cite{raonic2023convolutional} and DeepONet \cite{lu2021learning}, as the subnetworks in Figure \ref{Implicit}(a). These subnetworks are initialized with the same network configuration and may further employ a single subnetwork that shares parameters across different stages. While this approach might sacrifice some expressive power, it aligns more closely with mathematical property consistency and simplifies the learning process by reducing complexity. This parameter-sharing method will be proven to lead to more robust learning, thereby enhance the network’s ability to model complex phenomena accurately.

The implicit approach processes inputs of scatterer $q(x)$ and incident wave $u^i$ to directly output the scattered wave, primarily focusing on matching the final output with the ground truth rather than on the operational details of each subnetwork. All parameters are learned simultaneously. The Neumann series acts more as a guideline for network design than as a rigid protocol. This training approach has proven to enhance the capacity to solve practical problems effectively. Our previous research \cite{chen2024nsno} shows that the implicit approach offers a more accurate approximation than the Neumann series truncated at the same truncation number $L$. Additionally, the original Neumann series often demonstrates slow convergence or even divergence when faced with extremely high wavenumbers or large scatterer scales or magnitudes. In contrast, the implicit approach consistently yields superior performance.

\subsection{Explicit approach}
\label{Explicit approach}

\hspace{2em}The implicit approach sums the outputs of multi-layer subnetworks to produce the final output, focusing solely on achieving closeness to the ground truth without enforcing strict execution of each operator in the Neumann series by every subnetwork. Conversely, the explicit approach trains subnetworks to approximate each linear operator in the Neumann series, as discussed in this section.

As illustrated in Figure \ref{Explicit}(a), each subnetwork utilizes the input-output pairs ${(qu^{(j-1)}, u^{(j)})}_{j=1}^L$ from each equation of the Neumann series as training data. Considering the linear nature of the target operator $\hat{\mathcal{S}}$, employing a consistent network structure for these subnetworks, $\{\mathcal{N}_{{\theta}_j}\}, j=1, 2, \dots,$ is rational. Furthermore, as depicted in Figure \ref{Explicit}(b), a single unified network can approximate the operator, utilizing normalization operations from the implicit approach to accommodate inputs from different function spaces with variations in magnitude. This configuration also implies that multiple subnetworks share parameters, thereby maintaining consistency with their mathematical properties. Our experiments demonstrate that using only the initial equation pair ($qu^i,u^{(1)}$) can still yield excellent results. Data from subsequent equations can be employed to fine-tune the unified network.

\begin{figure}[h!]
\centering
	\includegraphics[width=0.8\textwidth]{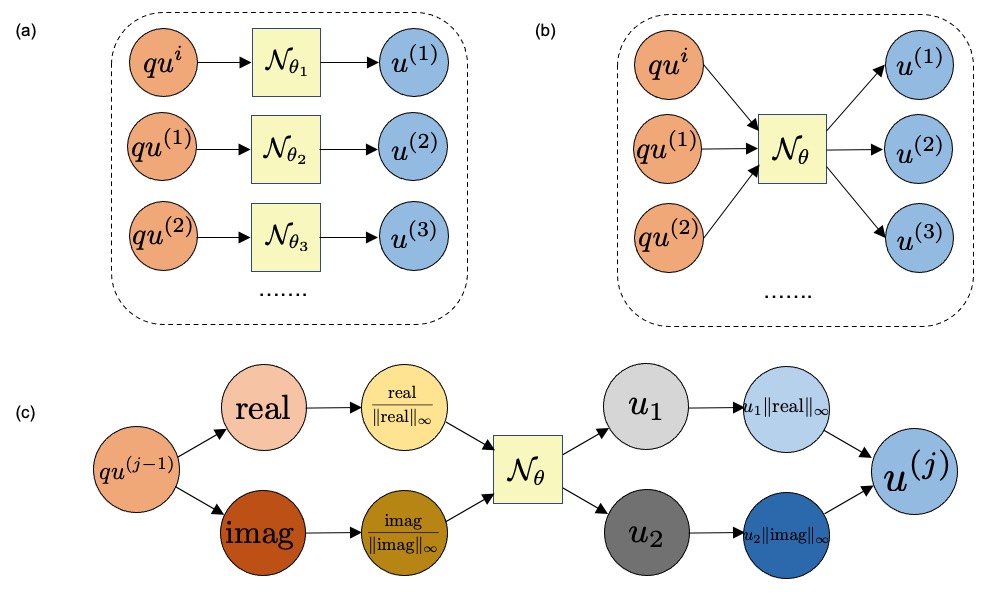}
	\caption{Training process for the explicit approach. \textnormal{$u_1,u_2$ are complex-valued network outputs, respectively,  $u^{(j)} = u_1 \| \text{real}\|_{\infty} + iu_2 \| \text{imag}\|_{\infty}, j = 1,2,\cdots,L, u^{(0)} = u^i.$}}
	\label{Explicit}
\end{figure}

Unlike the implicit approach, the explicit approach involves assembling trained subnetworks to calculate scattered waves based on the Neumann series. Although this linear looping operation reduces the potential for parallelism, it allows for the distinct separation of subnetwork training and assembly methods. On the one hand, this separation permits a focused effort on enhancing the precision of subnetworks in approximating the target operator properties. As depicted in Figure \ref{Explicit}(c), the real and imaginary parts of the inputs are processed separately by the network, producing complex-valued outputs that are subsequently recombined to form the final output. On the other hand, these finely trained subnetworks can be assembled not only through the Neumann series but also through any iterative method involving a common linear operator component. Theoretically, when the target series shows poor convergence or diverges, its counterpart, the explicit approach, also performs poorly. Thus, opting for more stable series can enhance overall computational accuracy. As a result, the explicit approach holds promise for achieving a customized balance between computational efficiency and precision.

\section{Simulation results for forward training}
\label{forward results}

\hspace{2em}To evaluate the effectiveness of Neumann series embedding, we choose FNO \cite{li2020fourier} as the baseline subnetwork, which has shown strong approximation ability in several two-dimensional parametric PDE problems. Below is a detailed explanation of the comparative methods covered in the paper.

\begin{itemize}
	\item \textbf{FNO:} We employ the original FNO to simply regard the input of the scatterer and the incident field as separate input channels.	
	\item \textbf{NSFNO:} In the implicit approach described in Section \ref{Implicit approach}, we substitute FNO structure for each subnetwork. Neumann series are integrated into the entire network.
	\item \textbf{FNONS:} In the explicit approach described in Section \ref{Explicit approach}, we substitute FNO structure for the small network. Neumann series are executed after the subnetwork is trained.
\end{itemize}

All proposed networks is implemented in PyTorch and the training is conducted on a NVIDIA GeForce RTX 3090 GPU card. In all subsequent experiments, unless explicitly specified, we will discretize the computational domain using a uniform grid of 129 $\times$ 129 points.  
 
\subsection{Training configurations}

\hspace{2em}The scatterer $q$ is sampled in the following smooth Gaussian-like distribution:
 
 \begin{equation}
 \begin{aligned}
 	 	q(x,y) &= \sum_{i=1}^\eta \lambda_i \exp( -a_i(x-b_i)^2 - c_i(y-d_i)^2), \qquad (x,y) \in \Omega_0 \\
 	 	\eta &\sim U(1,6) \quad \text{is a random integer}, \\
 	 	a_i,c_i &\stackrel{\text { i.i.d. }}{\sim} \text{Uniform}(R/2,R) 
 	 	\quad \text{for} \quad i=1,2,\dots,\eta, \\
 	 	b_i,d_i &\stackrel{\text { i.i.d. }}{\sim} \text{Uniform}(0.2,0.8)
 	 	\quad \text{for} \quad i=1,2,\dots,\eta,\\
 	 	\lambda_i &\stackrel{\text { i.i.d. }}{\sim} \text{Uniform}(-1,1)
 	 	\quad \text{for} \quad i=1,2,\dots,\eta,
 \end{aligned}	
 \label{q_gen}
 \end{equation}
 We set $R = 200$ to make sure the scatterer can be compactly supported. We then uniformly normalize the $L^{\infty}(q)$ to 0.1 to enhance training efficiency. As for the source, we only adopt 4 different plane waves, of which the incident angles are $0,{\pi}/{2},\pi,3\pi/2$. The ground truth of the corresponding scattered field are computed with the PML approach described in \ref{PML approach}. Specifically, The forward problem is converted into Helmholtz equation with complex coefficients and homogeneous boundary condition in a bounded domain,  then transformed into a linear system by the Finite Difference Method (FDM). The numerical simulation is first performed on a fine mesh with $513 \times 513$ grid points, then down sampled to a coarser mesh with $129\times 129$ grid points.

Except for the supervised data loss $\mathcal{L}_{data}$, we also introduce the physical informed PDE loss $\mathcal{L}_{PDE}$ to enhance the physical consistency and generalization capability. It is known that PINNs-like methods often fails to solve PDE with high frequency oscillation \cite{li2023physicsinformed}, so we combine two forms of loss with a hyperparameter $\lambda\in(0,1)$. The loss function can be summarized as follows: 

 \begin{equation}
 \label{Loss_eq}
 \begin{aligned}
 	 	\mathcal{L}_{total}(\mathcal{N}_{\theta}) &= \mathcal{L}_{data}(\mathcal{N}_{\theta}) + \lambda \mathcal{L}_{PDE}(N_{\theta}),\\
 	 	\mathcal{L}_{data}(\mathcal{N}_{\theta}) &= \mathbb{E}_{q,u^i} \|\mathcal{N}_{\theta}(q,u^i) - u^s\|_{L^2(\Omega)},\\
 	 	\mathcal{L}_{PDE}(\mathcal{N}_{\theta}) &=\left\{\begin{array}{l}
\frac{1}{k^2}\mathbb{E}_{q,u^i} \| \Delta \mathcal{N}_{\theta}(q,u^i) + k^2(1+q) \mathcal{N}_{\theta}(q,u^i) + k^2qu^i\|_{L^2(\Omega)} \;\text{ for FNO and NSFNO},\\
\frac{1}{k^2}\mathbb{E}_{q,u^i} \| \Delta \mathcal{N}_{\theta}(q,u^i) + k^2 \mathcal{N}_{\theta}(q,u^i) + k^2qu^i\|_{L^2(\Omega)} \;\text{   for FNONS,}
\end{array}\right.\\
\end{aligned}
 \end{equation}
where $\mathcal{N}_{\theta}$ denotes the network with parameters $\theta$. $u^s$ denotes the true scattered field. It should be pointed that the FNONS approach focuses on approximating the first linear PDE in the Neumann series \eqref{NS}, the ground truth differs from the others. 

As for detailed network structure, we make some modifications from the original FNO \cite{li2020fourier}. We replace the pointwise normalization step by a double-convolution layer to reinforce the feature extraction process. The Nerumann series in FNONS and NSFNO are both truncated with $L = 3$.

Our training set consists of 1024 instances of scatterer and the test set consists of 128 instances. The Adam optimizer \cite{kingma2014adam} is used to train the networks. 
Over 500 training epochs, the learning rate is adjusted using the cosine annealing scheduler \cite{loshchilov2016sgdr}, varying from 2e-3 to 1e-5, while maintaining a fixed batch size of 16.

\subsection{In-distribution performance}
\label{In}

\hspace{2em}In our hardware setting, the FNO model with 2.37 million parameters requires 8.85 seconds per training epoch. In contrast, the NSFNO model replicates the FNO's structure three times, tripling both the number of parameters and the training duration relative to FNO. Meanwhile, the FNONS configuration maintains a parameter volume akin to FNO's but incorporates dual computations for every complex value input, effectively doubling its training time in comparison to FNO. Notably, both FNO and NSFNO can directly give the scattered wave, FNONS necessitates an additional Neumann series summation after the training stage. 

\begin{table}[h!]
\begin{center}
\begin{tabular}{c|c|c|c}
      & k = 20 & k = 40 & k = 60 \\ \hline
FNO   &    0.14\%    &    0.40\%    &   1.20\%    \\ \hline
NSFNO &    0.03\%    &    0.20\%    &   0.60\%    \\ \hline
FNONS &    0.08\%    &    0.21\%    &   0.29\%   
\end{tabular}
\end{center}
\caption{Average relative error on the test set.}
\label{IDresults}
\end{table}

The average $L^2$ relative error of the four incident fields is summarized in Table \ref{IDresults}. We can observe that these methods all demonstrate similar performance on the test set, indicating their well generalization ability to  in-distribution scatterers. All over evaluated scenarios with wavenumber of 60, the FNO prediction exhibits a misfit of 1.20\%, which indicates the baseline accuracy among the models. The NSFNO model reduces the misfit to 0.60\%, demonstrating an enhanced fidelity. The FNONS further minimizes the misfit to 0.29\%, and reflects its superior approximation capability. It is evident that the incorporation of Neumann series (NS) significantly improves the accuracy when compared to the baseline FNO method. Specifically, FNONS method performs the best among the three methods.

We further select the incident wave directed from 0 degree with wavenumber $k=60$ to verify the discernible performance variation among different neural network methods. Figure \ref{example} illustrates this comparison: the first row displays the real part of the scatterer $q(x)$, the incident wave $u^i(x)$ and the precise scattered wave $u^s(x)$. The second row shows the real parts of the scattered waves as predicted by three distinct networks: FNO, NSFNO, and FNONS, respectively. The third row quantifies the discrepancy between these predictions and the exact solution. We can observe that that networks augmented with NS not only yield  more accurate, but also give smoother misfit, which is indicative of their improved stability and robustness in modeling complex wave phenomena.
  
\begin{figure}[h!]
	\centering
	\includegraphics[width=0.8\textwidth]{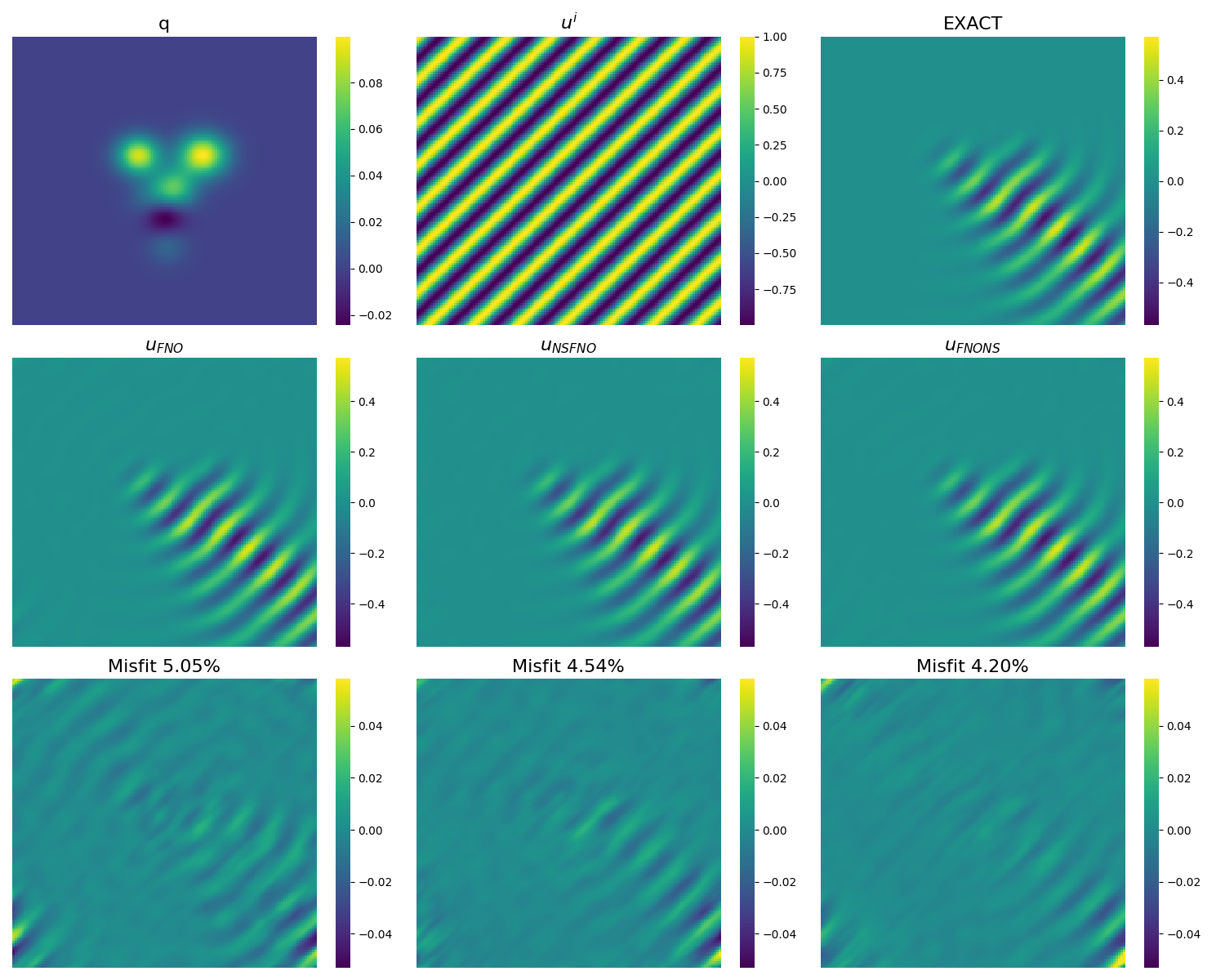}
	\caption{Comparison of predicted scattered fields at wavenumber k=60 for FNO, NSFNO, and FNONS.}
	\label{example}
\end{figure}

\begin{figure}[h!]
\centering
    \begin{subfigure}[h!]{0.3\linewidth}
        \centering
        \includegraphics[width=1\textwidth]{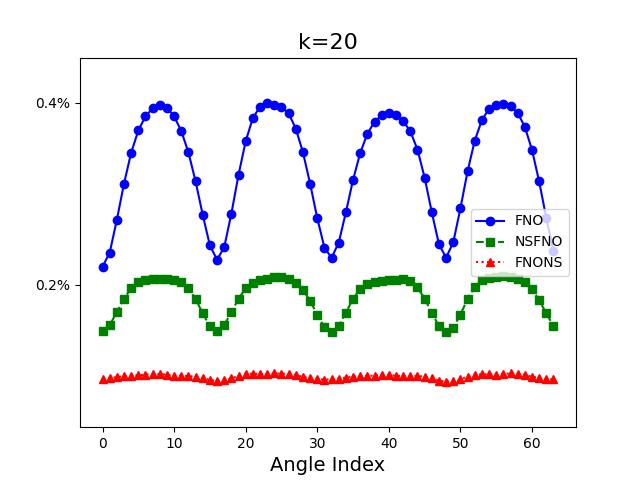}
    \end{subfigure}%
    \begin{subfigure}[h!]{0.3\linewidth}
        \centering
        \includegraphics[width=1\textwidth]{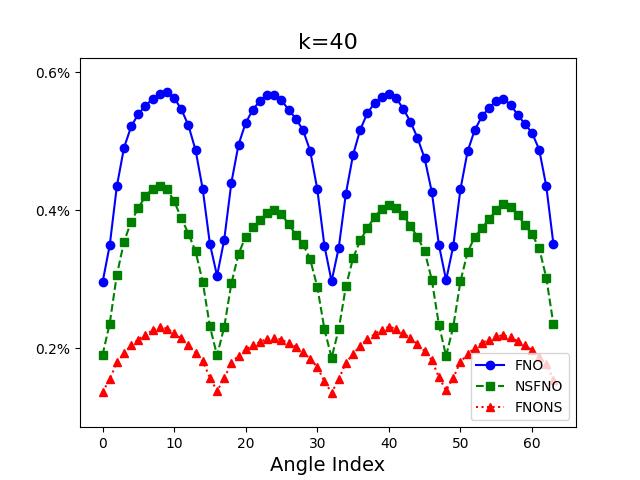}
    \end{subfigure}
        \begin{subfigure}[h!]{0.3\linewidth}
        \centering
        \includegraphics[width=1\textwidth]{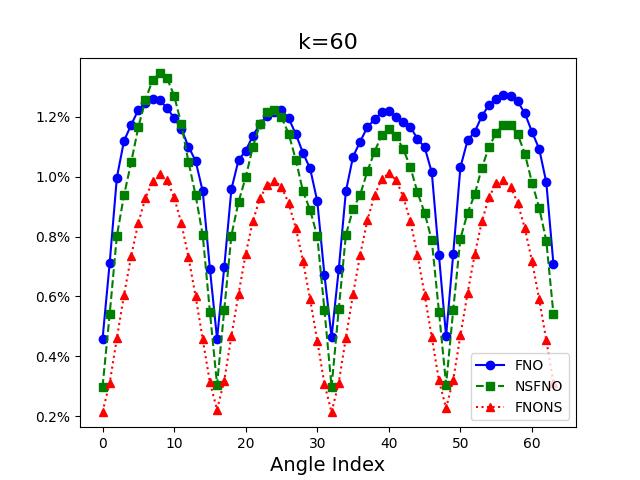}
    \end{subfigure}
    \caption{Network performance on out-of-distribution incident angles at various wavenumbers.}
    \label{generalization}
\end{figure}

\subsection{Out-of-distribution performance}

\hspace{2em}All three networks have shown good accuracy with in-distribution parameters. Specifically, the training is restricted to four incident angles to strengthen model robustness. Nevertheless, in the real-world inverse problem setups, incident waves originate from any direction, and it is crucial that the network reliably predicts the scattered wave for directions not included in the training set. Regarding scatterers, without prior information on their characteristics, we simply train the networks using smooth scatterers of uniform maximum magnitude of 0.1. However, real-world scatterers may greatly differ in magnitude, shape, and position from those in the training set. Furthermore, during inversion procedures, a series of scatterers are calculated iteratively and viewed as network inputs. Consequently, it is imperative to verify the network's effectiveness with varied scatterer inputs, confirming its utility as a reliable forward solver.

We first test generalization capability with respect to incident direction. For this purpose, we generate a set of 128 in-distribution scatterers and synthesized 64 plane waves coming from the directions $2\pi j/64,j = 0,1,\dots,63,$. Figure \ref{generalization} illustrates the average $L^2$ relative error metric to assess the difference between the predicted and actual scattered fields. The results indicate that all three networks maintain low error margins at angles included in the training set ($0,\pi/2,\pi,3\pi/2$), while also maintain controlled error at angles beyond the training distribution. Such results validate the effectiveness of our selective angle training approach. In the comparison of network performances, FNONS  consistently deliveres superior results over all considered wavenumbers and showcases remarkable robustness to variations in incident direction.

We further investigate networks' generalization capabilities by analyzing their performance on scatterers with various maximum magnitudes. In Figure \ref{magnitude_forward}, the relative errors with respect to exact scattered field are plotted, where the solid lines represent the case with incident angle 0 and the dashed lines represent the case with incident angle $\pi/4$. These angles correspond to the previously established lower and upper bounds of error for the networks at different incident angles. The FNO network shows a pronounced increase in error when the scatterer magnitude differs from that seen in the training dataset. However, integrating Neumann series mitigates this issue, maintains nearly unchanged relative error even for untrained magnitudes. This marked disparity in accuracy stems from the design of the network structure. The FNO model uses the scatterer function $q$ as an input channel. This could not address adequately on the solution's nonlinearity with respect to the scatterer properties. In contrast, the Neumann series transforms the nonlinearity into an approximate linearity within each sub-network, which substantially enhances the model's generalization capacity regarding scatterers $q(x)$.

Consequently, our experiments show that all networks reliably process in-distribution parameters, while incorporating Neumann series considerably enhances their capacity on generalization to out-of-distribution inputs. Detailed examination of networks dealing with complex scatterer configurations, including diverse shapes and positions, will be discussed in Section \ref{Inverse Simulation}.

\begin{figure}[h!]
\centering
    \begin{subfigure}[h!]{\linewidth}
        \centering
        \includegraphics[width=0.6\textwidth]{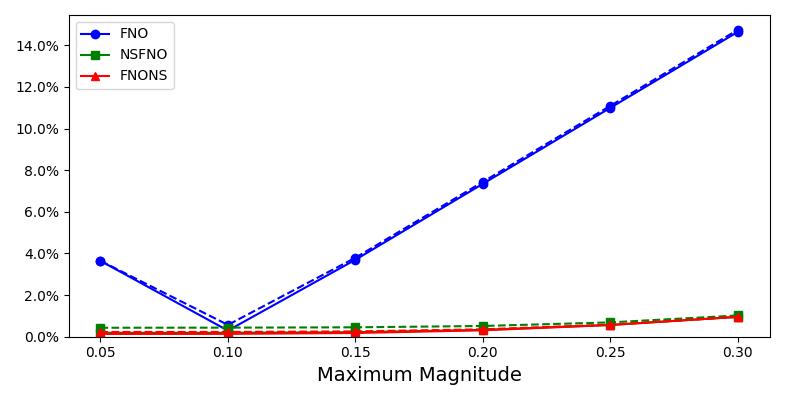}
    \end{subfigure}
    \caption{Relative error analysis of neural networks with different scatterer magnitudes at $k=40$.}
    \label{magnitude_forward}
\end{figure}

\section{Simulation results for inverse problem}
\label{Inverse Simulation}

\subsection{Inverse problem configuration}

\hspace{2em}We continue the classification method outlined in Section \ref{forward results} to compare performance of different forward solver. In addition to three network-based methodologies—FNO, NSFNO, and FNONS—we use the conventional finite difference method with MUMPS solver \cite{amestoy2000mumps} as the benchmark. In the following sections, we refer to it as MUMPS method. The finite difference method leads to a linear algebraic equation with a sparse matrix, which is more efficiently solved on the CPU. Forcing these methods onto the GPU requires converting the sparse matrix into a dense format, increasing memory usage and reducing computational efficiency. Consequently, the MUMPS method is executed on the CPU (Intel(R) Xeon(R) Gold 6226R CPU @ 2.90GHz) to leverage its efficiency with sparse matrices, while network-based methods are deployed on the GPU (RTX 3090) to utilize their computational power.

Regarding the simulation settings, unless specified otherwise, the experiments were conducted at the wavenumber $k=40$. We set $M=N=64 $ and discretize the computational domain with a 129 $\times$ 129 grid. To prevent the issue of the inverse crime,
measurements in our simulations are intentionally carried out on a finer mesh with 1025$\times$1025 grid points. During the optimization, we employ the L-BFGS algorithm \cite{liu1989limited} as the optimizer. The algorithm begins with a zero initial value, without requiring prior assumptions about the scatterer. The iterative process concludes when either it reaches a maximum within 15 iterations or any iteration exceeds 20 line search steps.
 
\subsection{Simple example}
\label{Simplest}

\hspace{2em}Following a similar manner to the forward problem, we initially select in-distribution scatterer to assess the feasibility of the inversion process. This specific target is characterized by

\begin{equation}
\label{test}
	q(x,y) = \exp(-150(x-0.3)^2-70(y-0.6)^2) - 0.7\exp(-40(x-0.7)^2 - 90(y-0.4)^2),
\end{equation}
We subsequently set the maximum value of $q$ at 0.1 and align it with the training stage. For simplicity, we initially operate within the framework of the adjoint state method (see Algorithm \ref{gradient}) and utilize noise-free measurements.

\remark{ Directly using the adjoint state method may introduce specific issues when involving network methods. Recall that the network serves as a surrogate solver for the forward problem, receiving inputs of scatterers and incident waves, aligning with the computation of $u_1$ in Algorithm \ref{algorithm}. However, when using the trained network to compute the adjoint equation $u_2$, the input distribution differs from that of the forward problem. This discrepancy may lead to inaccuracies in gradient calculations. We will explore the impact of this issue in Section \ref{ability}.}
\label{remarkk}

\begin{figure}[h!]
\centering
    \begin{subfigure}[h!]{0.35\linewidth}
        \centering
        \includegraphics[width=\textwidth]{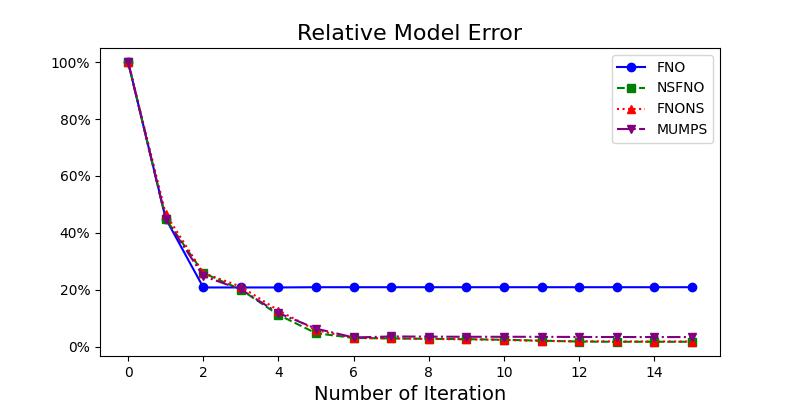}
        \caption{}
    \end{subfigure}%
    \begin{subfigure}[h!]{0.35\linewidth}
        \centering
        \includegraphics[width=\textwidth]{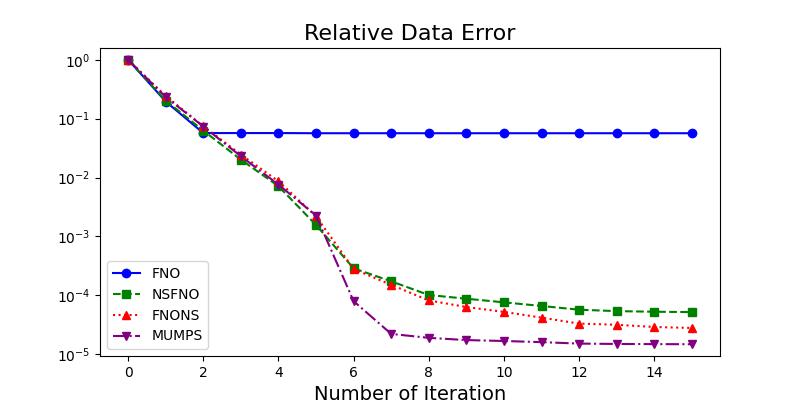}
        \caption{}
    \end{subfigure}
    \begin{subfigure}[h!]{\linewidth}
        \centering
        \includegraphics[width=0.7\textwidth]{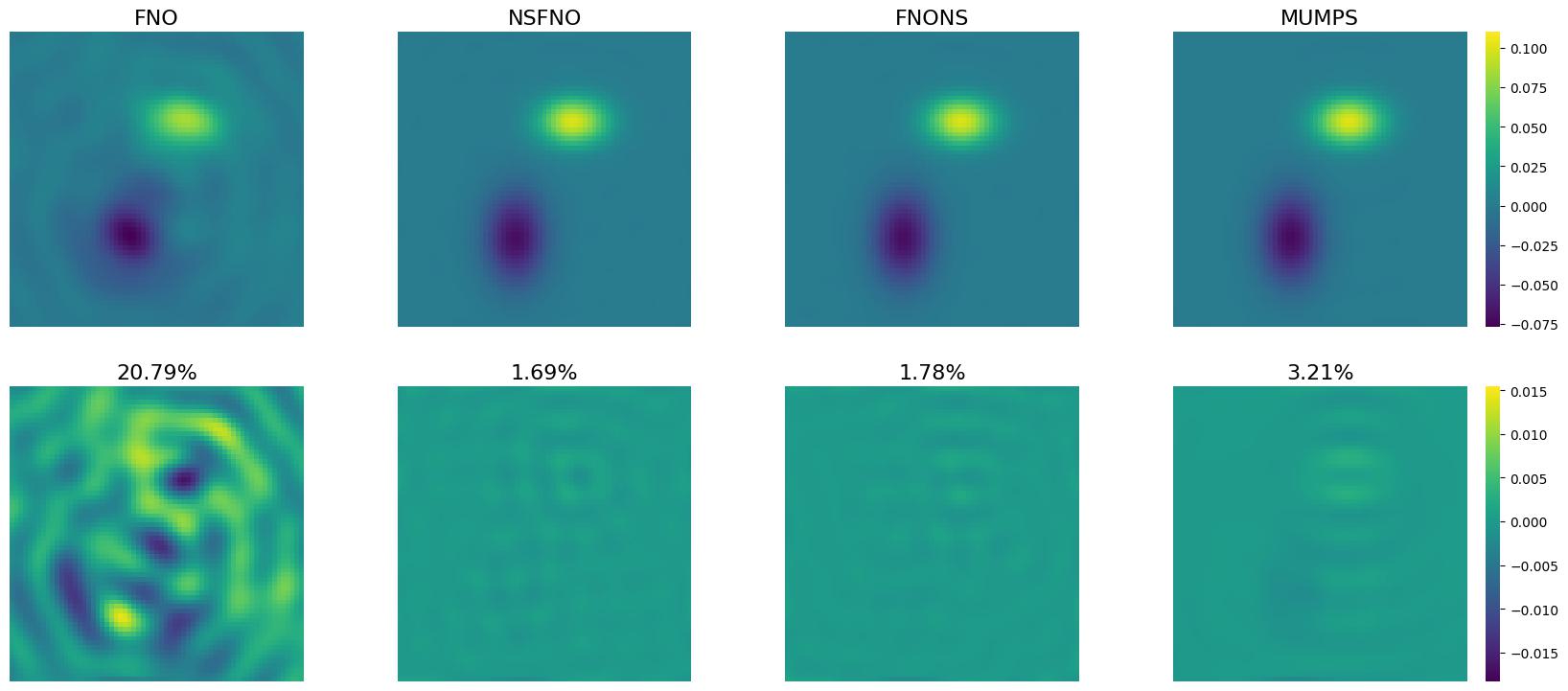}
        \caption{}
    \end{subfigure}
    \caption{Performance comparison of different inversion approaches at $k=40$. (a) shows evolution of the relative $L^2$ error between the reconstructed scatterer and the exact scatterer, (b)  tracks evolution of the relative objective function value, (c) presents best reconstruction results and the corresponding relative error percentage.}
    \label{k40}
\end{figure}

Figure \ref{k40} compares the performance of four approaches at $k=40$. The FNO approach captures the basic profile of the scatterer but concludes iterations prematurely due to lower solution precision. In comparison, NSFNO and FNONS perform comparably and match the MUMPS approach in terms of model and data errors. This indicates that integrating the Neumann series with FNO significantly enhances the accuracy of the inversion solution by improving its performance on the forward problem. Further details on computational efficiency will be presented in the following section.

\subsection{Performance across different magnitudes}
\label{Magnitude}

\hspace{2em}Let us revisit the initial motivation for using network-based approaches. We aim to train a high-precision surrogate network model for the forward problem. The trained neural network allows for multiple calls during the inversion process and maintains accuracy while significantly speeds up the inversion. In this framework, the training process can be conducted offline before being called upon, operating independently from the online iterative optimization process. We anticipate that the surrogate network model will accurately handle inputs from various scatterers without prior knowledge of their characteristics, including their profiles and magnitudes. Therefore, in addition to the in-distribution tests discussed in the previous section, we must also evaluate the model’s ability to successfully reconstruct out-of-distribution scatterers.

In this section, we examine the influence of scatterer's magnitude. We first retain the scatterer from the previous section but adjust its magnitude to 0.4. Figure \ref{0.4-ASM} reveals that FNO reaches its expressive limit early in the iterations and results in an imprecise outcome. In contrast, FNONS consistently reduces the objective function value, significantly outperforms NSFNO and produces results comparable to MUMPS. These results align closely with those from the forward problem tests in Figure \ref{Inverse Simulation}, where the FNO method performs poorly when tasked with recovering scatterers of larger magnitudes.

\begin{figure}[h!]
\centering
    \begin{subfigure}[h!]{0.35\linewidth}
        \centering
        \includegraphics[width=\textwidth]{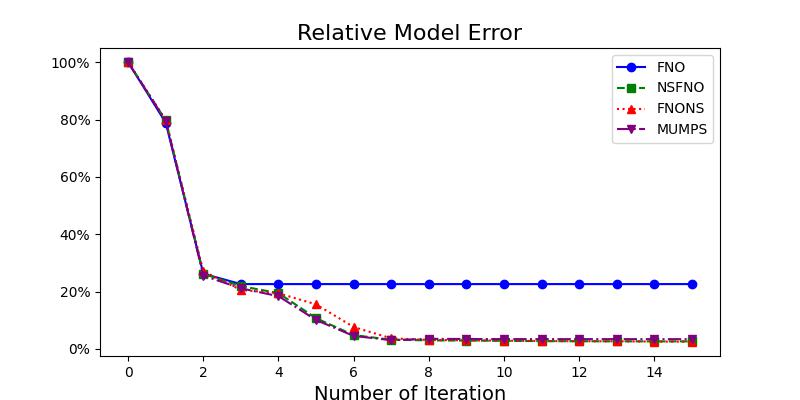}
    \end{subfigure}%
    \begin{subfigure}[h!]{0.35\linewidth}
        \centering
        \includegraphics[width=\textwidth]{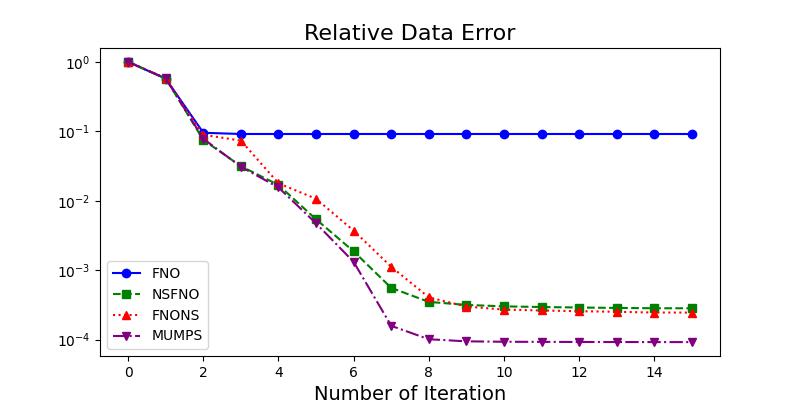}
    \end{subfigure}
        \begin{subfigure}[h!]{0.7\linewidth}
        \centering
        \includegraphics[width=\textwidth]{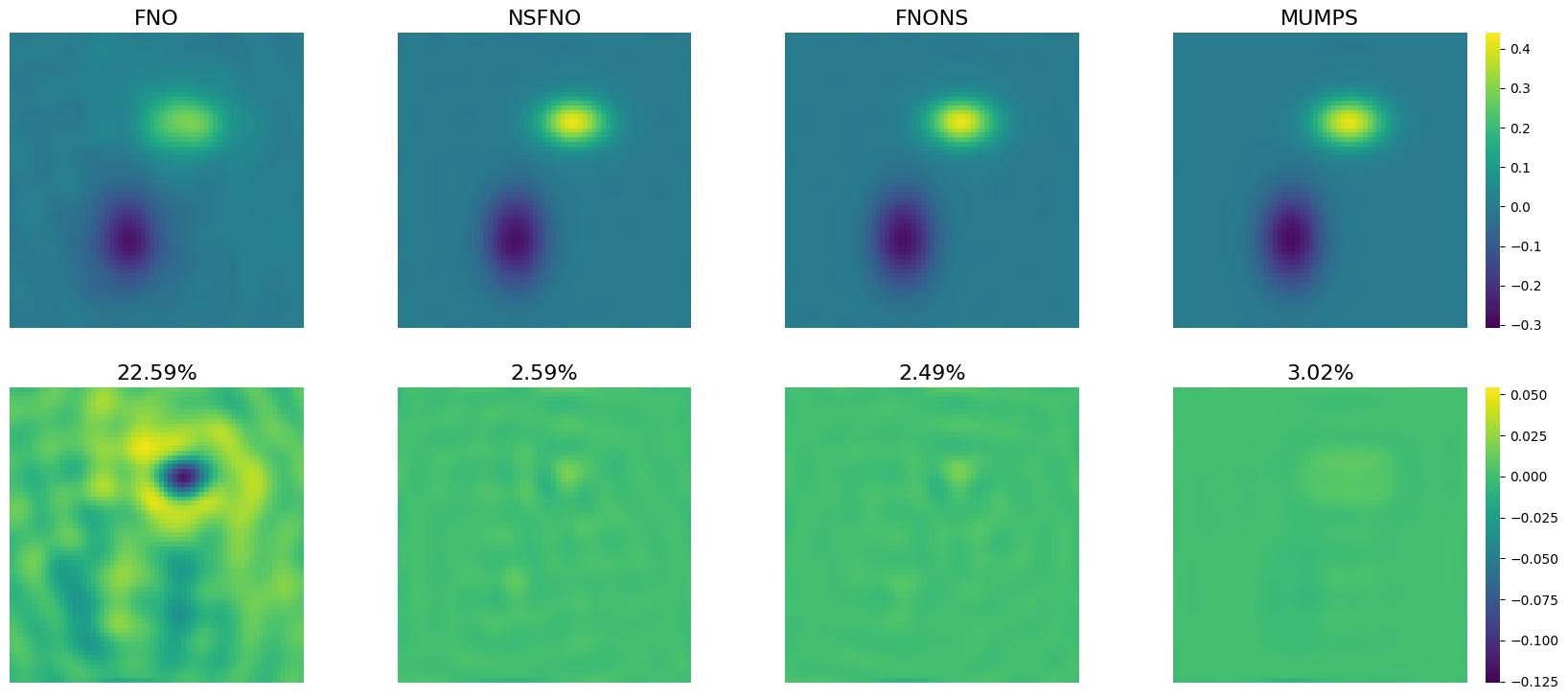}
    \end{subfigure}
    \caption{Performance comparison of different inversion approaches using the adjoint state method at $k=40$ and $\|q\|_{\infty} = 0.4$.}
    \label{0.4-ASM}
\end{figure}

Besides using the Adjoint State Method (ASM) for gradient computation, automatic differentiation (AD) offers an alternative by systematically breaking down mathematical expressions into elementary operations and applying the chain rule \cite{paszke2017automatic}. This method avoids the need for deriving adjoint equations when neural network serves as solver for the forward problem and is more flexible. We exhibit the inversion performance by using AD to compute the gradients in Figure \ref{0.4-AD}, which demonstrates that AD enhances the efficiency and accuracy of inversion processes for NSFNO and FNONS under the same settings as ASM. This improvement suggests that AD's potential to improve inversion accuracy, especially with complex network-based models, is significant. ASM may suffer from inaccurate gradient calculations as noted in Remark \ref{remarkk}, while gradient computation by AD is intrinsically linked to the network itself, which makes the accuracy of the inverse problem heavily dependent on the precision of the forward problem. This distinction underlines AD's advantage in achieving more reliable inversion results. Figure \ref{two_compare} further validates this perspective. The relative error of the FNO remains consistent across various magnitudes under ASM, as shown in (a). In comparison, as depicted in (b), automatic differentiation allows the relative error of inversion to exhibit a trend similar to that observed in the forward problem in Figure \ref{Inverse Simulation}.

\begin{figure}[h!]
\centering
    \begin{subfigure}[h!]{0.35\linewidth}
        \centering
        \includegraphics[width=\textwidth]{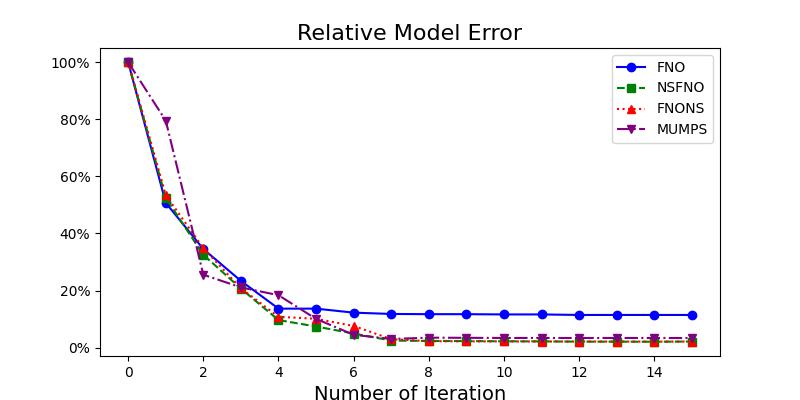}
    \end{subfigure}%
    \begin{subfigure}[h!]{0.35\linewidth}
        \centering
        \includegraphics[width=\textwidth]{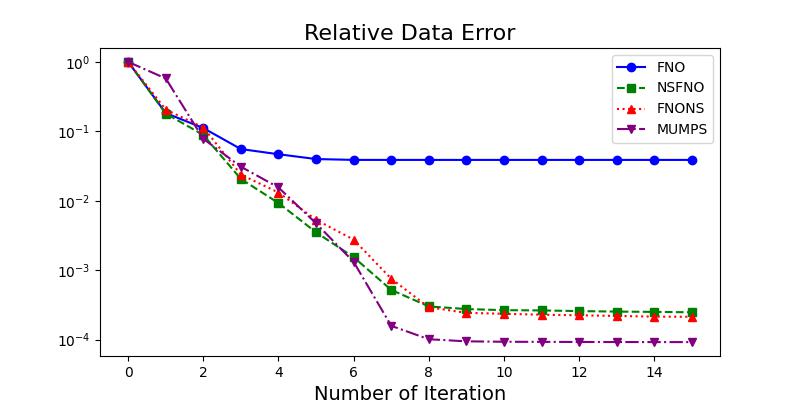}
    \end{subfigure}
        \begin{subfigure}[h!]{0.7\linewidth}
        \centering
        \includegraphics[width=\textwidth]{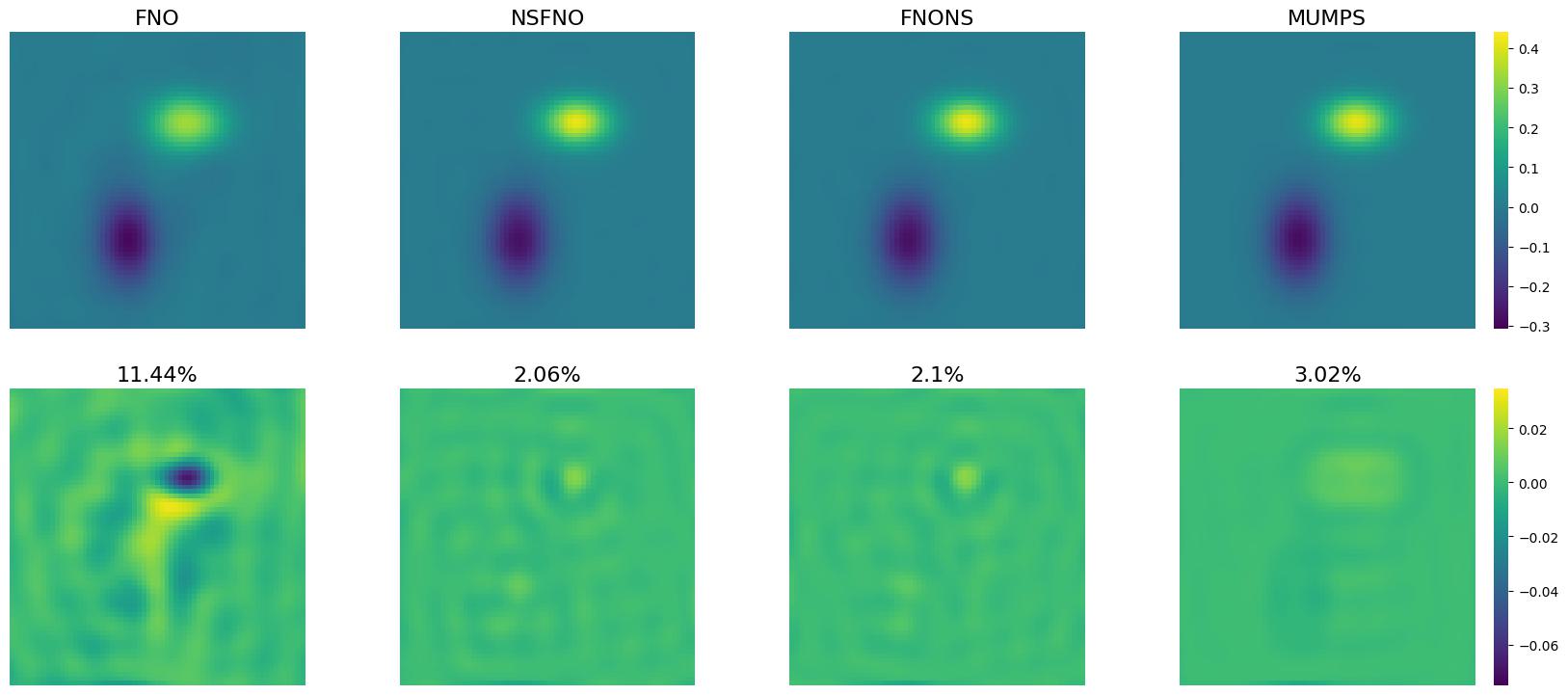}
    \end{subfigure}
    \caption{Performance comparison of different inversion approaches using automatic differentiation at $k=40$ and $\; \|q\|_{\infty} = 0.4$.}
    \label{0.4-AD}
\end{figure}

\begin{figure}[h!]
	\centering
	\begin{subfigure}[h!]{0.4\linewidth}
	\includegraphics[width = \textwidth]{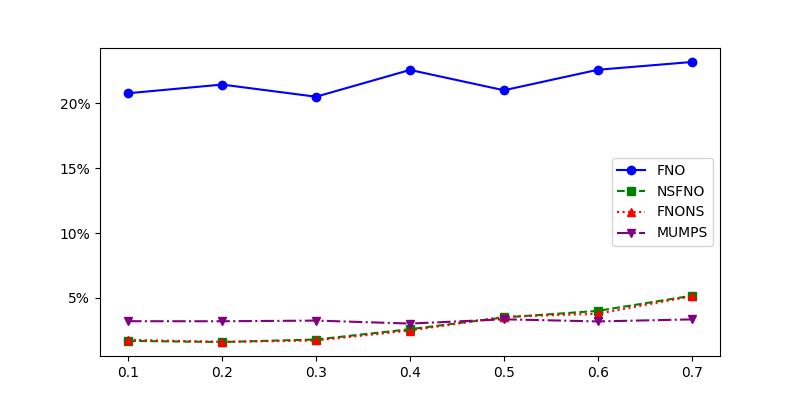}
	\caption{Adjoint state method}
	\end{subfigure}
		\centering
	\begin{subfigure}[h!]{0.4\linewidth}
	\includegraphics[width = \textwidth]{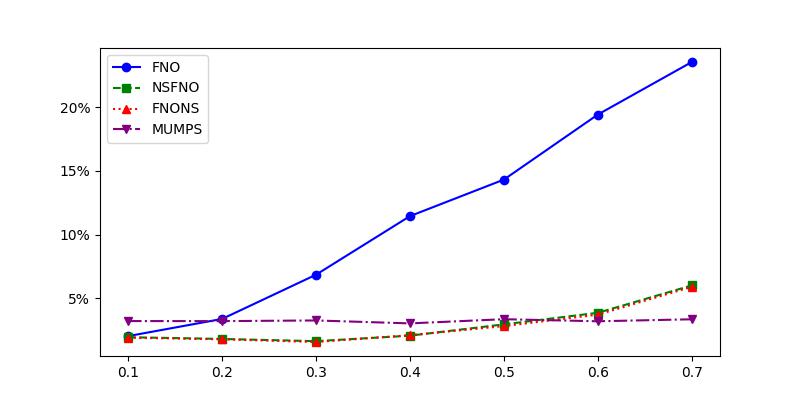}
	\caption{Auto differentiation}
	\end{subfigure}
	\caption{The evolution of the relative reconstruction error in relation to the scatterer's scale for automatic differentiation. The horizontal axis represents the magnitude of the scatterer and the vertical axis represents the relative error of the recovered scatterer.}
\label{two_compare}
\end{figure}

\begin{table}[h!]
\centering
\begin{tabular}{cccccc}
\hline
 Gradient Solver& \multicolumn{1}{c}{Approach} &  N-fev/N-grad & T-avg(s) & T-toal(s) \\ \hline
 \multicolumn{1}{c}{\multirow{4}{*}{Automatic Differentiation}}
 & \textbf{FNO} &  56 & 0.083 & 4.690 \\
 \multicolumn{1}{c}{}
 & \textbf{NSFNO} & 20 & 0.184 & \textbf{3.689} \\
 \multicolumn{1}{c}{}
 & \textbf{FNONS} & 20 & 0.339 & 6.799 \\ 
 \hline
\multicolumn{1}{c}{\multirow{4}{*}{Adjoint State Method}} 
& \textbf{FNO} &  89 & 0.051 & 4.582 \\
\multicolumn{1}{c}{} 
& \textbf{NSFNO}  & 18 & 0.122 & \textbf{2.200} \\
\multicolumn{1}{c}{} 
& \textbf{FNONS}  & 16 & 0.226 & 3.617 \\

  \multicolumn{1}{c}{}
 & \textbf{MUMPS} & 17 & 4.205 & 71.49 \\
 \hline
\end{tabular}

\caption{Comparison of the computational efficiency of different gradient solvers. The computation of both $J(q)$ and its gradient $\nabla J(q)$ is considered a single computational process. The term 'N-fev/N-grad' refers to the aggregate count of these computational processes performed. T-avg represents the average runtime for function evaluations. T-total is the product of these two, representing the overall running time.} 

\label{effiency}
\end{table}

Table \ref{effiency} demonstrates the computational efficiency of gradient solvers for various methods at the magnitude 0.4. Both ASM and AD, particularly when enhanced with Neumann series (NS), align the number of function evaluations with the standard 15 iterations, benefit from robust forward problem's performance across different magnitudes. Specifically, NSFNO shows superior performance, with total times of 2.200 seconds under ASM and 3.689 seconds under AD, and significantly outperforms the benchmark numerical method, MUMPS, which requires 71.49 seconds. This improvement is substantial and demonstrates the efficiency of NS-enhanced methods. Generally, ASM processes gradients more rapidly than AD, correlating to a single forward network evaluation compared to AD's backpropagation.

\subsection{Performance across different geometries}
\label{ability}

This section focuses on the inversion related to scatterers with increasingly complex out-of-distribution profiles. Four distinct targets were selected for testing: (a) overlapping spheres, (b) a large rectangular robot figure, (c) Austria, and (d) multiple small-scale scatterers. These samples, markedly different from the smooth Gaussian scatterers in the training set, encompass challenges such as large support sets, sharp shapes, multi-scale features, and intricate microstructures. Moreover, to further test the limits of our NS-enhanced networks during the inversion process, the magnitude of the scatterers was increased to 0.6.

\begin{table}[h!]
\centering
\begin{tabular}{ccccccccc}
\hline
 & & \multicolumn{2}{c}{\textbf{FNO}} & \multicolumn{2}{c}{\textbf{NSFNO}} & \multicolumn{2}{c}{\textbf{FNONS}} & \textbf{MUMPS} \\
 \hline
& & ASM & AD & ASM & AD & ASM & AD & ASM \\
 \hline
 \multicolumn{1}{c}{\multirow{4}{*}{\textbf{Rel-err}}} 
 &\textbf{(a)} & 51.11\% & 39.05\% & \textbf{26.60\%} & 27.95\% & 26.94\% & 28.03\% & 24.20\% \\
 &\textbf{(b)} & 57.34\% & 49.37\% & \textbf{34.18\%} & 34.99\% & 35.02\% & 35.08\% & 24.91\% \\
 &\textbf{(c)} & 65.98\% & 56.84\% & 44.47\% & \textbf{38.71\%} & 41.91\% & 38.92\% & 23.72\% \\
&\textbf{(d)} & 56.35\% & 48.70\% & \textbf{39.43\%} & 42.31\% & 39.56\% & 42.57\% & 35.24\% \\
 \hline
  \multicolumn{1}{c}{\multirow{4}{*}{\textbf{SSIM}}}
&\textbf{(a)} & 0.6092 & 0.7381 & 0.7910 & 0.8106 & 0.7934 & \textbf{0.8146} & 0.8272 \\
&\textbf{(b)} & 0.2874 & 0.4915 & 0.5853 & 0.6282 & 0.5897 & \textbf{0.6368} & 0.7377 \\
&\textbf{(c)} & 0.2682 & 0.3778 & \textbf{0.5355} & 0.4633 & 0.4433 & 0.4797 & 0.7423 \\
&\textbf{(d)} & 0.4438 & 0.5841 & 0.5945 & \textbf{0.6539} & 0.6044 & 0.6413 & 0.7034 \\
\hline
\end{tabular}
\caption{Misfits of the reconstructed scatterers with different geometries.}
\label{geo}
\end{table}

\begin{figure}[h!]
\centering
\includegraphics[width=0.8\textwidth]{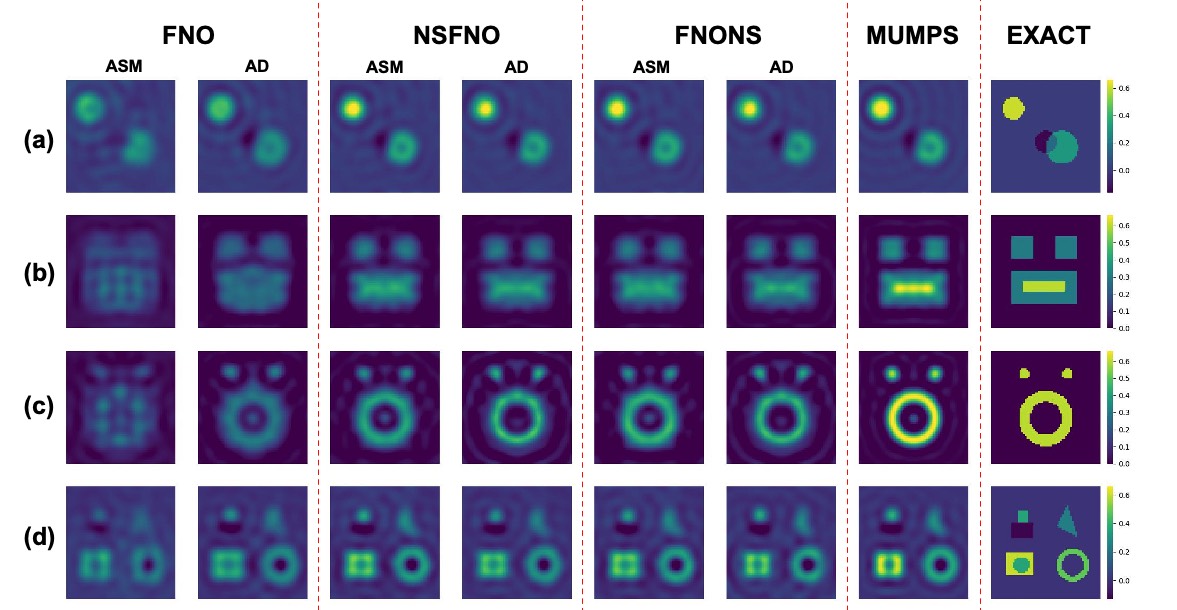}
\caption{Reconstructions of scatterers with different geometries. ASM denotes the adjoint state method to compute the gradient and AD denotes the automatic differentiation method.}
\label{scatterers}
\end{figure}

Figure \ref{scatterers} presents the reconstruction results. 
MUMPS, serving as the benchmark, consistently performs the best. Notably, since MUMPS generates the supervised data for network training, it generally represents the upper limit of accuracy achieved by network methods. The NS-enhanced methods (NSFNO and FNONS) consistently outperform the basic FNO across various scatterer shapes and gradient solvers (ASM and AD). For the lower precision FNO, ASM is a preferable choice; however, for NS-enhanced networks, both gradient solvers yield reconstructions closely approximating those produced by MUMPS. Table \ref{geo} compares the relative errors and Structural Similarity Index Measure (SSIM) performance and highlights the effectiveness of NS-enhanced approaches.

\begin{figure}[h!]
\centering
\includegraphics[width = 0.6\textwidth]{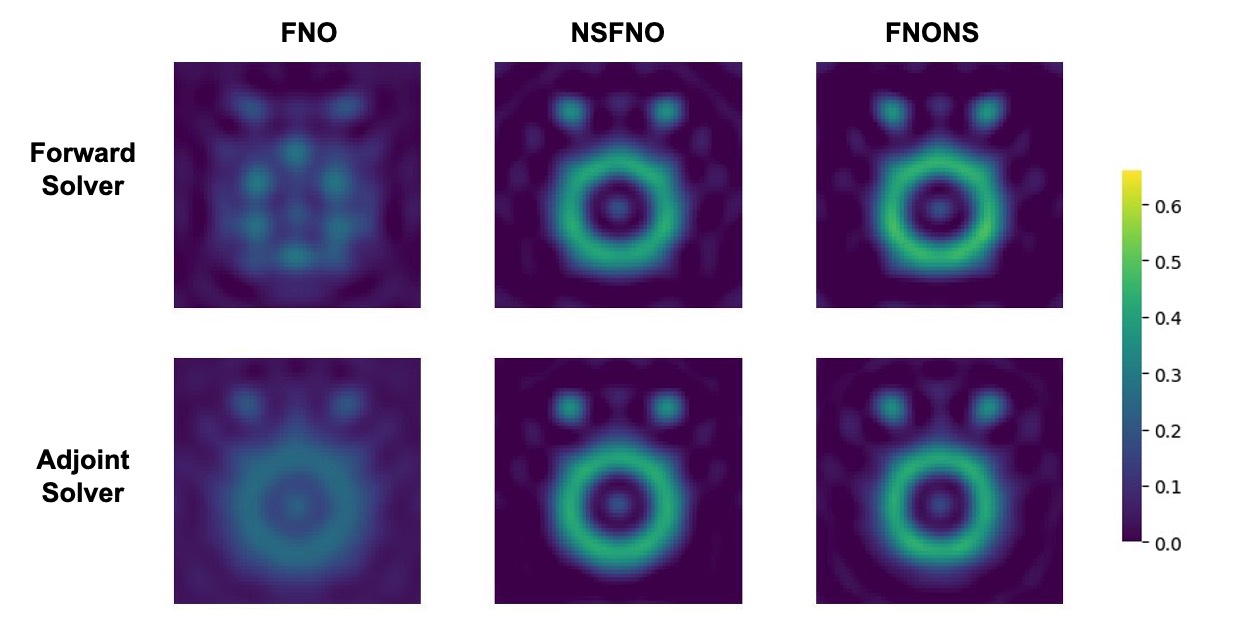}
\caption{Reconstructed results of different gradient solvers with the adjoint state method.}
\label{error-example4}
\end{figure}

Given the verification of the NS-enhanced model's inversion capabilities on such "out-of-distribution" scatterers, we address the concern posed in Remark \ref{remarkk}. This remark questions the impact on reconstruction accuracy when the network processes inputs outside its training distribution, specifically during the computation of the adjoint state equation in ASM. To establish a baseline, we train a network with the same architecture, referred to as the Adjoint solver.
Figure \ref{error-example4} demonstrates that, except for FNO, the reconstructions by the two NS-enhanced methods are remarkably similar. Integrating these results for complex reconstruction scenarios with those in Figure \ref{two_compare}, it is evident that, when using network methods, particularly NS-enhanced, there is no need to train a separate network for the adjoint equation.

\subsection{Robustness of reconstruction}
\label{robustness}

\hspace{2em}Real-world applications often face data with noise due to environmental disturbances or inherent measurement errors. Assessing the robustness of reconstruction algorithms is essential for their practical utility, especially in fields such as medical imaging and geophysical exploration where precise image reconstruction is critical.

\begin{figure}[h!]
\centering
\begin{subfigure}[h!]{0.55\linewidth}
\includegraphics[width = \textwidth]{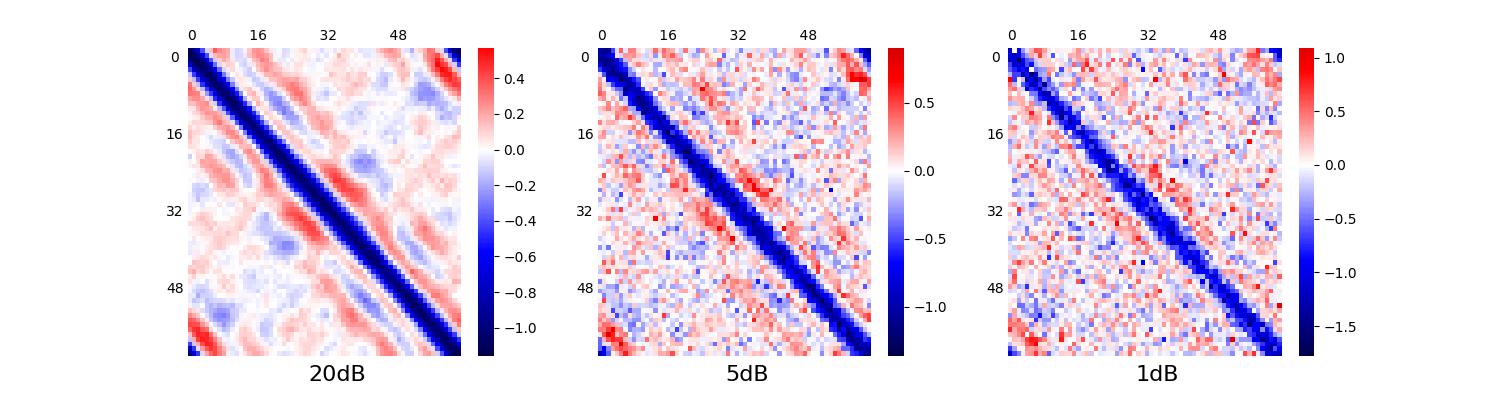}
\caption{}
\end{subfigure}
\begin{subfigure}[h!]{0.4\linewidth}
\includegraphics[width = \textwidth]{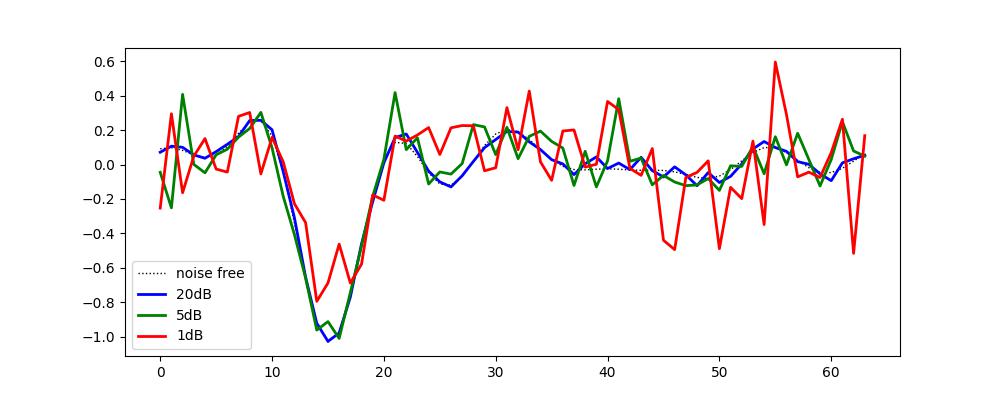}
\caption{}
\end{subfigure}
    \caption{Received scattering data for different Signal-to-Noise Ratios (SNR). (a) shows the real part of the received signals, with the horizontal and vertical axes representing the indices of transmitters and receivers positioned around a circle, respectively. (b) The incident wave approaches from the direction of $\theta = \pi/4$. In the graph, the horizontal axis represents the real part of the received scattered wave, while the vertical axis denotes the locations of the receivers.}
\label{noisy data}
\end{figure}

We select the scatterer profile type (c) from the previous section for testing and set its magnitude to 0.5. Figure \ref{noisy data} illustrates how noise impacts the accuracy of the received signals. Notably, at noise levels of 5dB and 1dB, the data is significantly disrupted, which poses substantial challenges for reconstruction. Figure \ref{error-example1} illustrates the comparative performance of various reconstruction methods at different noise levels. MUMPS consistently delivers the most accurate reconstructions across all noise levels and we establish the reconstructions by MUMPS as the benchmark. As noise level is increased to 5dB and 1dB, the basic FNO exhibits a noticeable decline in performance since it generates increased artifacts and diminished clarity in the images. In comparison, NSFNO and FNONS demonstrate significant robustness and maintain closer fidelity to the original image. They significantly outperform FNO, particularly at higher SNR. 

\begin{figure}[htbp]
\centering
\includegraphics[width=0.7\textwidth]{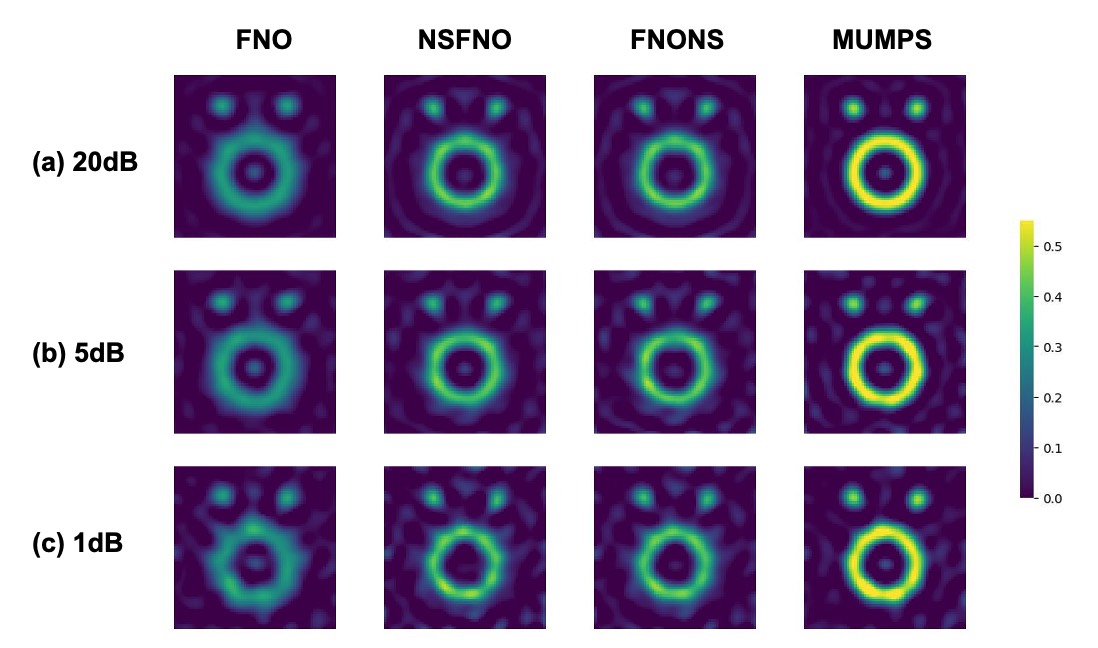}
\caption{Reconstructions of the scatterer for different SNR.}
\label{error-example1}
\end{figure}

The comparison across different SNRs underscores the enhanced robustness of NSFNO and FNONS compared to the FNO method. These enhanced models consistently deliver reconstructions close to the MUMPS benchmark results, 
which establishes them as reliable options for practical applications where noise is prevalent.

\subsection{Performance across various transmitter-receiver layouts}
\label{layout}

\hspace{2em}In practical inverse medium problems, the configuration of transmitters and receivers is crucial. In applications such as seismic inversion \cite{alkhalifah2014scattering} or medical diagnostics \cite{kuprat2009anisotropic}, operational constraints determine the placement of transmitters and receivers, which directly impacts data quality and the effectiveness of inversion algorithms. 
Therefore,
evaluating various layouts aids in assessing the cost-effectiveness of different configurations in industrial control systems.

\begin{figure}[h!]
\centering
\includegraphics[width = 0.6\textwidth]{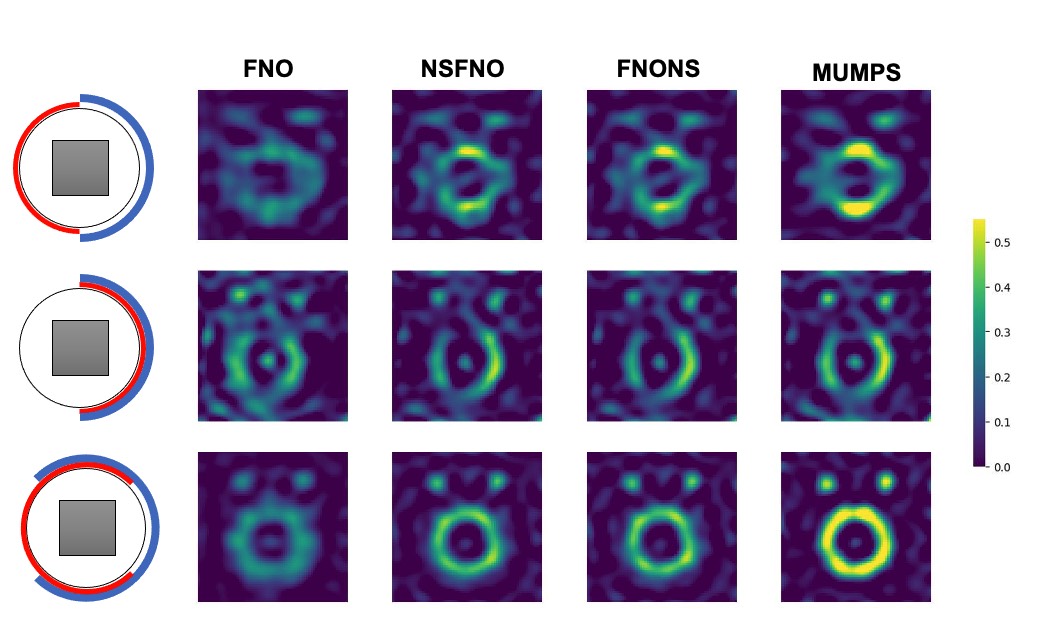}
\caption{
Reconstructions of the scatterer across various transmitter-receiver layouts. Red and blue circles respectively indicate the positions of transmitters and receivers.}
\label{error-example}
\end{figure}

In this section, we focus on a limited-angle setting for the same scatterer as in the previous section and impose a noise with an SNR of 5dB on the data. The reconstruction results are displayed in Figure \ref{error-example}. In the first two configurations, due to the non-optimal placement of transmitters and receivers, even the MUMPS method results in inaccuracies and artifacts. By contrast, the NS-enhanced networks largely match the MUMPS reconstructions, while the basic model appears overly blurred. The third configuration represents the minimal transmitter and receiver settings required for numerical methods to approximately recover the scatterer's outline. In this scenario, both FNONS and NSFNO deliver results significantly surpassing those of FNO.

The assessment of various layouts confirms that NS-enhanced models can match numerical results in practical settings. Their robustness qualifies them as viable and efficient alternatives for designing adaptable and effective systems, where achieving a balance between scientifically robust solutions and economic viability is essential.

\subsection{Performance in high-wavenumber scenarios}
\label{high-wavenumber}

\hspace{2em}In various practical applications, such as medical imaging, geophysical exploration, and industrial non-destructive testing, high reconstruction accuracy is crucial. Common techniques to enhance the precision include regularization strategies and iterative refinement \cite{haber2000optimization}. The process typically initiates with low-frequency inversions to approximate the global minimum, followed by high-frequency inversions that refine the reconstruction and delineate finer details. Therefore, testing the network's capability at higher wavenumbers is also essential to evaluate its suitability as a alternative for forward problem solver.

We select scatterer profile type (d) from the previous section to better examine the inversion capabilities for multiple small-scale scatterers and preserve the 5dB noise level in last section. Figure \ref{bigk} displays visual comparisons of the reconstruction outcomes. The NSFNO and FNONS models demonstrate closer alignment with the MUMPS results, whereas the basic FNO model performs poorly, and produce blurred results, particularly at $k=60$.

\begin{figure}[h!]
\centering
\includegraphics[width=0.7\textwidth]{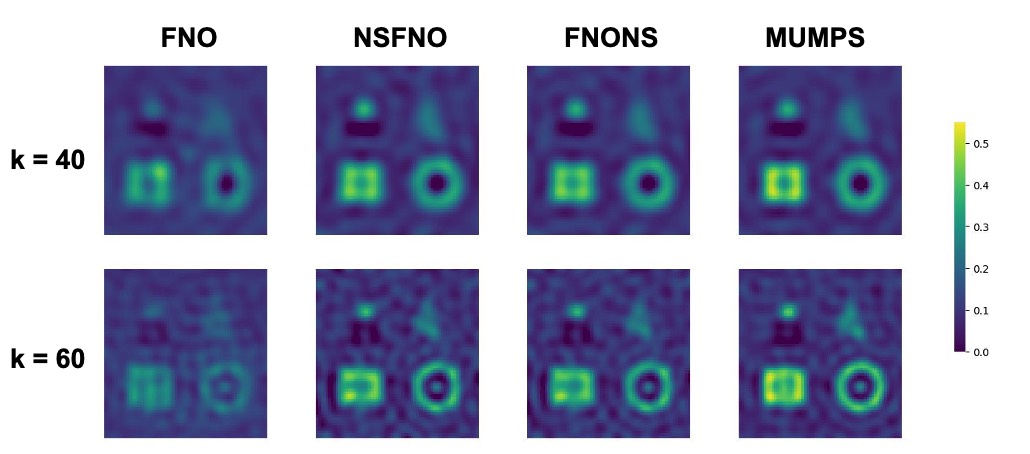}
\caption{Reconstructions of the scatterer at various wavenumbers.}
\label{bigk}
\end{figure}

Table \ref{bigk1} provides a quantitative assessment. At the two wavenumbers $k=40$ and $60$, both the NS-enhanced networks demonstrate lower relative errors and significantly improved SSIM compared to the FNO results. As the wavenumber increases, the recovery quality of the FNO model deteriorates and contradicts the goal of higher precision. Although the NS-enhanced methods achieve higher recovery accuracy, their relative errors exhibit smaller variations compared to the benchmark MUMPS method, which indicates a potential reduction in the network's solving capability.

\begin{table}[h!]
\centering
\begin{tabular}{cccccc}
\hline
& & \textbf{FNO} & \textbf{NSFNO} & \textbf{FNONS} & \textbf{MUMPS} \\
 \hline
\multicolumn{1}{c}{\multirow{2}{*}{\textbf{Rel-err}}}
& k=40 & \textbf{51.36\%} & 40.07\% & 40.05\% & 36.55\%  \\
& k=60 & 59.13\% & \textbf{37.18\%} & \textbf{38.04\%} & \textbf{31.11\%} \\
\hline
\multicolumn{1}{c}{\multirow{2}{*}{\textbf{SSIM}}}  
& k=40 & \textbf{0.5073} & 0.6296 & 0.6305 & 0.6609  \\
& k=60 & 0.4615 & \textbf{0.6336} & \textbf{0.6467} & \textbf{0.6782} \\
\hline
\end{tabular}
\caption{Misfit of the reconstructed results at different wavenumber.}
\label{bigk1}
\end{table}

To the best of the authors' knowledge, solving high-frequency Helmholtz equations with neural networks \cite{song2022versatile} remains a formidable challenge. Higher wavenumbers necessitate computations on finer grids to produce more accurate training data, thus require further refinement of the network architecture to address this intricate learning task. Nevertheless, NS-enhanced networks have explored the feasibility of high-frequency alternative solutions.

\subsection{Adaptability across various network structure}
\label{Adaptability}

\hspace{2em}In the previous subsections, we thoroughly evaluated the enhancements that the NS-enhanced method brings to the FNO model.  Given its specific architectural constraints and varying performance across different scenarios, focusing exclusively on the FNO model might introduce limitations. Therefore, to assess the generalization ability of the NS-enhanced method, it is essential to explore its effectiveness across a variety of neural network architectures. 

To assess the generalization ability of the NS-enhanced method, in addition to the FNO model, we have incorporated UNO (A fourier-enhanced UNet \cite{ronneberger2015u}, detailed in \cite{chen2024nsno}) and CNO \cite{raonic2023convolutional}, which have recently shown promising results in the field of operator learning, as base models. For the scatterers, we created a MNIST-like dataset consisting of ten samples for each digit from 0 to 9, with magnitudes set in the range of 0.2 to 0.7. Inversions are conducted at a wavenumber of k=40, with the data subjected to 5dB of noise.

Figure \ref{mnist} displays one of the reconstruction results from three base models: FNO, UNO, and CNO, where "NS-" and "-NS" represent implicit and explicit embedding methods, respectively. Across all models, embedding Neumann series consistently demonstrates closer approximations to the exact and MUMPS results, particularly in terms of clarity and detail preservation. Notably, the NS-enhanced versions of each model show improved definition and sharper features in the digit outlines compared to their original versions. 

\begin{figure}[h!]
\centering
\includegraphics[width=0.7\textwidth]{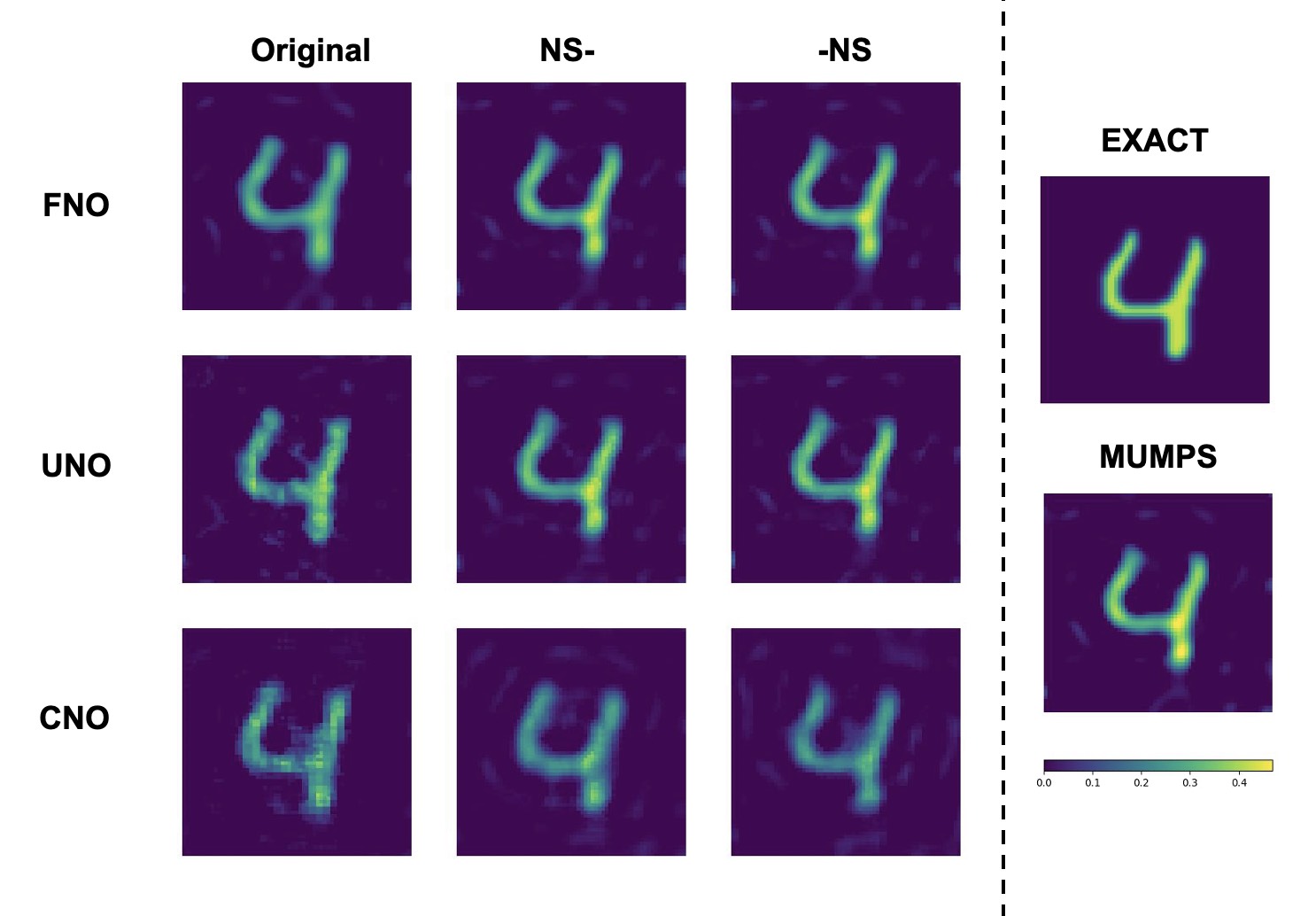}
\caption{Sample of reconstructions on MNIST.}
\label{mnist}
\end{figure}

\begin{table}[h!]
\centering
\begin{tabular}{cccccc}
\hline
Approach & \textbf{Rel-err} & \textbf{SSIM} & \textbf{N-fev} & \textbf{T-avg(s)} & \textbf{T-total(s)} \\
\hline
\textbf{FNO} & 40.80\% & 0.8138 & 74.77 & 0.133 & 9.97 \\
\textbf{NSFNO} & \textbf{26.76\%} & \textbf{0.8827} & 16.74 & 0.308 & \textbf{5.15} \\
\textbf{FNONS} & 26.80\% & 0.8826 & 17.62 & 0.652 & 11.50 \\
\hline
\textbf{UNO} & 45.20\% & 0.8079 & 64.90 & 0.130 & 8.45 \\
\textbf{NSUNO} & 28.48\% & 0.8642 & 18.35 & 0.312 & \textbf{5.73} \\
\textbf{UNONS} & \textbf{27.36\%} & \textbf{0.8798} & 17.44 & 0.671 & 11.71 \\
\hline
\textbf{CNO} & 48.67\% & 0.7988 & 29.27 & 0.132 & \textbf{3.87} \\
\textbf{NSCNO} & \textbf{35.94\%} & \textbf{0.8503} & 16.51 & 0.302 & 5.00 \\
\textbf{CNONS} & 44.71\% & 0.8132 & 17.35 & 0.889 & 15.44 \\
\hline
\textbf{MUMPS} & 14.95\% & 0.9548 & 16.11 & 4.212 & 67.87 \\
\hline
\end{tabular}
\caption{Accuarcy and efficiency comparison on MNIST.}
\label{mnist1}
\end{table}

Table \ref{mnist1} presents a detailed comparison across various neural network models. The NS-enhanced models significantly outperform their original counterparts in terms of relative error and SSIM, and demonstrate superior image reconstruction quality. While the NS-enhanced method results in a slower average solving time (T-avg) due to the stacking of the original networks, it reduces the number of function evaluations (N-fev), thereby achieving comparable total times (T-total). In comparison with MUMPS, the NS-enhanced models exhibit a slight decrease in accuracy but achieve a 5 to 15-fold reduction in total computation time, significantly enhancing computational efficiency.

As a result, the NS-embedding approach has been proven effective for training multiple-input forward networks and addressing their corresponding inverse problems across various neural network architectures. This plug-and-play feature facilitates straightforward integration into existing state-of-the-art computational frameworks, significantly expanding its application scope.

\section{Conclusion}

\hspace{2em}This study presents a novel network framework designed to address the inverse medium problem and capable of processing diverse input types. By incorporating the Neumann series into both the network and post-calculation phases, we develop two variants: the implicit and explicit networks. These proposed methods significantly improve computational efficiency without compromising accuracy, surpassing the original network in both speed and reliability. Specifically, our networks demonstrate robust generalization capabilities, adeptly handling variations in the shapes and magnitudes of scatterers. Additionally, the framework shows strong resilience against noise and limited data, providing rapid and dependable initial solutions for practical applications. Future research will focus more on the performance of models in high-wavenumber regimes and in the presence of large scatterers. Moreover, the universal design of this approach facilitates its adaptation to a broad spectrum of inverse medium challenges, including those related to obstacles, sound-hard surfaces, and transmission media.

\appendix
\section{Adjoint state method}
\hspace{2em}A widely-adopted approach to derive the gradient expression is based on the following adjoint state method \cite{plessix2006review}. With this method, to derive the gradient, we write the Lagrange functional as:

\begin{equation*}
\begin{aligned}
\mathcal{L} (u,\lambda,q) &= \operatorname{Re}(\tilde{J}(u) - \langle \lambda,F(u,q)),\\
\text{where} \quad   \tilde{J}(u) &= \frac{1}{2} \| T u-\bm{d_j}\|_2^2,\\
F(u,q) &= u - \mathcal{S}(q)(qu_j^i),\quad F(u^{\text{opt}},q)=0,\\
\langle f,g\rangle &= \int_{\Omega} f\bar{g} \,dx,
\end{aligned}
\end{equation*}

Let the derivatives of $\mathcal{L}$ with respect to $u$ equals 0, we get the following adjoint-state equation
\begin{equation*}
\begin{aligned}
    \left(\frac{\partial F(u^{\operatorname{opt}},q)}{\partial u}\right)^* \lambda^{\operatorname{opt}} = \frac{\partial \tilde{J}(u^{\operatorname{opt}})}{\partial u},\\
\Rightarrow \lambda^{\operatorname{opt}} =T^{*}(Tu^{\operatorname{opt}}-\bm{d_j}),
\end{aligned}
\end{equation*}
where $(u^{\operatorname{opt}},\lambda^{\operatorname{opt}})$ denotes the optimal solution to $\mathcal{L}(u,\lambda,q)$. The gradient of $J$ can be computed as follows:
\begin{equation*}
\begin{aligned}
	\frac{\partial J_j(q)}{\partial q} &= \frac{\partial \mathcal{L}(u,\lambda,q)}{\partial q}\Bigg|_{(u,\lambda) = (u^{\text{opt}},\lambda^{\text{opt}})}
	= (\frac{\partial \tilde{J}(u)}{\partial q} - \frac{\partial \operatorname{Re} \langle \lambda, F(u,q)\rangle}{\partial q})\Bigg|_{(u,\lambda) = (u^{\text{opt}},\lambda^{\text{opt}})},
\\
	& = - \operatorname{Re}((\frac{\partial F(u,q)}{\partial q} )^*\lambda) \Bigg|_{(u,\lambda) = (u^{\text{opt}},\lambda^{\text{opt}})}
	=\operatorname{Re} (( \frac{\partial \mathcal{S}(q)(qu_j^i)}{\partial q})^* \lambda^{\text{opt}}).
\end{aligned}
\end{equation*}

Furthermore, the Fr\'{e}chet derivative can be computed by
the following variational derivation
\begin{equation*}
\begin{aligned}
    & \quad\mathcal{S}(q+\delta q)( (q + \delta q)) u_j^i) - \mathcal{S}(q)( q u_j^i) \\
& = \mathcal{S}(q+\delta q)(\delta q u_j^i) + \mathcal{S}(q+\delta q)( q u_j^i) - \mathcal{S}(q)( q u_j^i)\\
& = \mathcal{S}(q+\delta q)(\delta q u_j^i)+ \mathcal{S}(q)(\delta q \mathcal{S}(q+\delta q)( q u_j^i)),\\
&\Rightarrow \frac{\partial \mathcal{S}(q)(qu_j^i)}{\partial q} \delta q = \mathcal{S}(q)((u_j^i+\mathcal{S}(q)(qu_j^i)) \delta q),
\end{aligned}    
\end{equation*}
taking the above expression back to the expression for the gradient, we can get
\begin{equation*}
\begin{aligned}
	\frac{\partial J_j(q)}{\partial q} &= \operatorname{Re} (( \mathcal{S}(q)((u_j^i+\mathcal{S}(q)(qu_j^i) )\cdot ))^* \lambda^{\text{opt}})\\
	&= \operatorname{Re} (( \overline{  u_j^i+\mathcal{S}(q)(qu_j^i)   })(\mathcal{S}^*(q)T^{*}(T \mathcal{S}(q)(qu_j^i)-\bm{d_j})),
\end{aligned}
\end{equation*}
where $(u_j^i+\mathcal{S}(q)(qu_j^i) )\cdot $ denotes the pointwise multiplication operator. Furthermore, the following theorem simplifies the process to compute the adjoint operator $\mathcal{S}^*(q)$.
\begin{theorem}
\label{theorem}
$\mathcal{S}^*(q) = \overline{\mathcal{S}(q)}$.
\end{theorem}

\begin{proof}
Given $ f,g \in C_0(\Omega)$, we can define the corresponding forward map
\begin{equation*}
	u = \mathcal{S}(q)f, \, v = \mathcal{S}(q)\overline{g}.
\end{equation*}

 Then by introducing the Hermitian inner product and substitute \eqref{LS}, we get
\begin{equation*}
\begin{aligned}
\langle \mathcal{S}(q)f,g \rangle 
&= \langle u,g\rangle 
= \langle \hat{\mathcal{S}}(qu)+\hat{\mathcal{S}}(f),g\rangle, \\
\langle f,\overline{\mathcal{S}(q)}g \rangle 
&= \langle f, \overline{v} \rangle 
= \langle f, \overline{\hat{\mathcal{S}}(qv)}+\overline{\hat{\mathcal{S}}(\overline{g})} \rangle ,
\end{aligned}
\end{equation*}
thereby, we have
\begin{equation*}
\begin{aligned}
\langle \hat{\mathcal{S}}(qu),g \rangle 
&= \langle \hat{\mathcal{S}}(\overline{g}),\overline{qu} \rangle= \langle v - \hat{\mathcal{S}}(qv) , \overline{qu} \rangle,\\
&= \int_{\Omega} q(x)u(x)v(x)  \, dx - \int_{\Omega} \int_{\Omega} G(x,y)q(x)v(x)q(y)u(y) \,dxdy\\
&= \langle u-\hat{\mathcal{S}}(qu),\overline{qv}\rangle = \langle \hat{\mathcal{S}}(f) ,\overline{qv} \rangle =\langle f, \overline{\hat{\mathcal{S}}(qv)}\rangle,\\
\langle \hat{\mathcal{S}}(f),g \rangle 
&= \int_{\Omega} \int_{\Omega} G(x,y) f(y) \overline{g(x)} \, dx dy 
= \langle f,\overline{\hat{\mathcal{S}}(\overline{g})} \rangle,
\end{aligned}
\end{equation*}
hence,
\begin{equation*}
	\langle \mathcal{S}(q)f,g \rangle = \langle f,\overline{\mathcal{S}(q)}g \rangle 
	\quad \Rightarrow \quad \mathcal{S}^*(q) = \overline{\mathcal{S}(q)}.
\end{equation*}
\end{proof}

As a result, we can summarize the procedure for computing the gradient as follows:

\begin{algorithm}

\caption{Adjoint state method}
\label{algorithm}
\begin{algorithmic}[1]

  \State \textbf{Input:} scatterer $q$, incident waves $u_1^i,u_2^i,\cdots,u_M^i$  
  \State \textbf{Output:} The gradient of $J$ with respect to $q$
  \For{$j = 1,2,\dots M$}
    \State $u_1 \gets \mathcal{S}(q) (qu_j^i),$
    \State $u_2 \gets \mathcal{S}(q) (\overline{T^* (T u_1 -\bm{d_j})}),$
    \State $\nabla_q (J_j) \gets \operatorname{Re} ( u_2 (u_j^i+u_1)),$ 
  \EndFor
  \State $\nabla_q (J) \gets \nabla_q (J_1) + \nabla_q (J_2) + \cdots \nabla_q (J_M).$
  \State \Return $\nabla_q (J)$

\end{algorithmic}
\label{gradient}
\end{algorithm}

\section{Implementation for forward problem}
\label{PML approach}

\begin{figure}[h!]
\centering
	\includegraphics[width=0.3\textwidth]{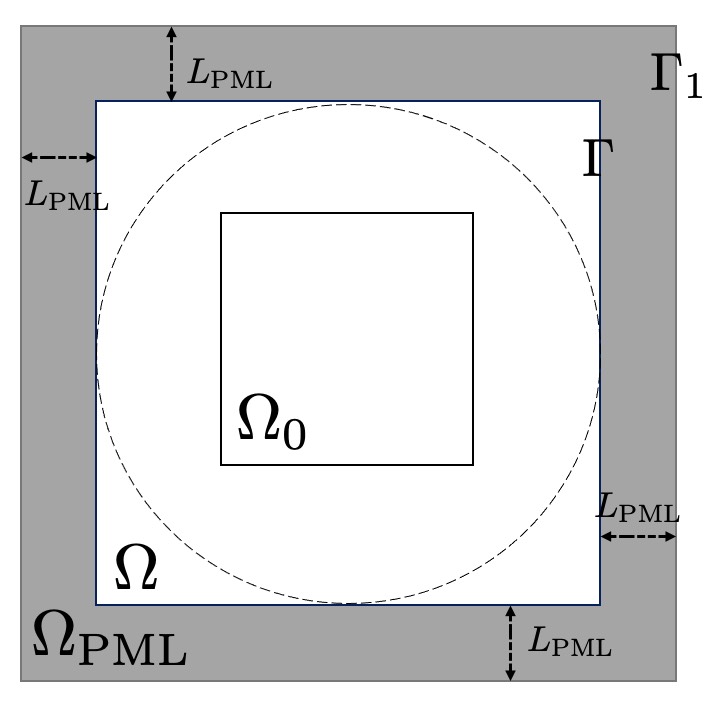}
	\caption{A schematic diagram illustrating the PML approach. $L_{\text{PML}}$ denotes the thickness of the gray PML region. $\Omega_0$ denotes the domain of interest, $\Omega$ denotes the area of computation.}
	\label{PML}
\end{figure}

As depicted in Figure \ref{PML}, the gray region, known as the Perfectly Matched Layer (PML), surrounds the physical domain and efficiently absorbs outgoing waves, thereby enabling the simulation of boundary absorption conditions at infinity within a finite area. We reference \cite{heikkola2003fast, kang2011inverse} to present the following heterogeneous Helmholtz equation:

\begin{equation}
\left\{
\begin{aligned}
	\frac{\partial}{\partial x}\left(\frac{e_y}{e_x} \frac{\partial u_{\text{PML}}}{\partial x}\right)+\frac{\partial}{\partial y}\left(\frac{e_x}{e_y} \frac{\partial u_{\text{PML}}}{\partial y}\right)+e_x e_y k^2(1+q) u_{\text{PML}} + k^2qu^i &=0,\; \text{in}\;\overline{\Omega}\cup \Omega_{\text{PML}} \\
	u_{\text{PML}} &= 0, \; \text{on} \;\Gamma_1,
\end{aligned}
\right. 
\label{PML_eq}
\end{equation}
where
$$
e_x:= \begin{cases} 1- ik \left(\frac{x}{L_{\text{PML}}}\right)^2, & -L_{\text{PML}}<x\leq0 \\ 1, & 0<x\leq 1\\
1- ik \left(\frac{x-1}{L_{\text{PML}}}\right)^2, & 1<x<1+L_{\text{PML}}.\end{cases}
$$
A similar definition applies to $e_y$, and $L_{\text{PML}}$ is set as 0.05. As a result, $u_{\text{PML}}$ provides an effective approximation for $u^s$ within the region $\Omega$.


\end{document}